\documentclass[reqno]{amsart}
\usepackage{bm}
\usepackage{tikz}
\usetikzlibrary{decorations.pathmorphing,patterns}
\usetikzlibrary{patterns.meta}
\usepackage{caption,subcaption}
\usepackage{epstopdf, epsfig}
\usepackage{float}
\usepackage{bbm}
\usepackage[leqno]{subeqnarray}
\usepackage[margin=1in]{geometry}
\usepackage{algorithm,algorithmic}
\newtheorem{theorem}{Theorem}[section]
\newtheorem{conj}{Conjecture}[section]
\newcommand{\Weber}{\operatorname{\mathit{W\kern-.30em e}}}
\newcommand{\Rey}{\operatorname{\mathit{R\kern-.20em e}}}
\newcommand{\Capil}{\operatorname{\mathit{C}}}

\usepackage{hyperref}

\makeatletter
\newsavebox{\@brx}
\newcommand{\llangle}[1][]{\savebox{\@brx}{\(\m@th{#1\langle}\)}%
  \mathopen{\copy\@brx\mkern2mu\kern-0.9\wd\@brx\usebox{\@brx}}}
\newcommand{\rrangle}[1][]{\savebox{\@brx}{\(\m@th{#1\rangle}\)}%
  \mathclose{\copy\@brx\mkern2mu\kern-0.9\wd\@brx\usebox{\@brx}}}
\makeatother

\title{Optimal control of thin liquid films and transverse mode effects}

\begin{document}

\author{R. J. Tomlin, S. N. Gomes, G. A. Pavliotis
\and
  D. T. Papageorgiou
}

\maketitle

\begin{abstract}
We consider the control of a three-dimensional thin liquid film on a flat substrate, inclined at a non-zero angle to the horizontal. Controls are applied via same-fluid blowing and suction through the substrate surface. We consider both overlying and hanging films, where the liquid lies above or below the substrate, respectively. We study the weakly nonlinear evolution of the system, which is governed by a forced Kuramoto--Sivashinsky equation in two space dimensions. The uncontrolled problem exhibits three ranges of dynamics depending on the incline of the substrate: stable flat film solution, bounded chaotic dynamics, or unbounded exponential growth of unstable transverse modes. We proceed with the assumption that we may actuate at every point on the substrate. The main focus is the optimal control problem, which we first study in the special case that the forcing may only vary in the spanwise direction. The structure of the Kuramoto--Sivashinsky equation allows the explicit construction of optimal controls in this case using the classical theory of linear quadratic regulators. Such controls are employed to prevent the exponential growth of transverse waves in the case of a hanging film, revealing complex dynamics for the streamwise and mixed modes. We then consider the optimal control problem in full generality, and prove the existence of an optimal control. For numerical simulations, we employ an iterative gradient descent algorithm. In the final section, we consider the effects of transverse mode forcing on the chaotic dynamics present in the streamwise and mixed modes for the case of a vertical film flow. Coupling through nonlinearity allows us to reduce the average energy in solutions without directly forcing the linearly unstable dominant modes.
\end{abstract}

\section{Introduction}

The dynamics of thin liquid films in various physical situations form a range of classical and important problems in fluid mechanics, attracting a lot of attention from researchers in the field (see \cite{kalliadasis2011falling,craster2009dynamics} and the references therein). In this work, we study the optimal control problem for a gravity-driven, thin, viscous liquid film on an inclined flat substrate, where controls are applied at the substrate surface by means of same-fluid blowing and suction. We allow the fluid to be either overlying or hanging, where the film lies above or below the substrate, respectively. Falling liquid films have a wide range of industrial applications, including coating processes \cite{doi:10.1146/annurev.fl.17.010185.000433,doi:10.1146/annurev.fluid.36.050802.122124}, and heat and mass transfer \cite{FRISK19721537,shmerler1986effects,LYU19911451,Miyara1999,serifi2004transient}. For coating processes, a stable film with a flat interface is desirable, whereas heat (and mass) transfer is improved by the presence of interfacial deformations; this is due to increased surface area, the presence of recirculation regions in the wave crests, and effective film thinning where the heat transfer becomes almost purely conductive. Thus, it is desirable to know how to control a liquid film to have a desired predetermined interface.

In the absence of controls, the system possesses an exact flat film solution \cite{nusselt1} (known as the Nusselt solution) with a semi-parabolic velocity profile in the streamwise direction. The experimental work of Kapitza and Kapitza \cite{kapitza1949wave} displayed the vast range of dynamical behaviours that could be obtained with flows of this kind, and promoted analytical and numerical studies of the problem. For an overlying film, Yih \cite{yih1955proceedings,Yih1} and Benjamin \cite{FLM:367246} considered the linear stability of the exact Nusselt solution, and showed that the film is unstable to long waves beyond a critical Reynolds number which depends on the angle of inclination -- for vertical arrangements, the critical Reynolds number is zero. Starting with the full Navier--Stokes equations, families of reduced-order models may be constructed to simplify the problem both analytically and numerically, with the aim to capture the evolution of the film interface. With a long-wave assumption and formal asymptotics, a Benney equation \cite{Benney,gjevik1970occurrence} for the interface height may be constructed for Reynolds numbers close to critical -- such equations usually retain the effects of gravity, viscosity and surface tension. These highly nonlinear models have been studied extensively and generalised by a number of authors \cite{lin1969finite,roskes1970three,nakaya1975long,atherton1976derivation}, however lack global existence of solutions and finite-time blow-ups have been observed in numerical simulations \cite{pumir1983solitary,joo1991long,rosenau1992bounded}. Although including many physical mechanisms, Benney equations are not effective at modelling thin film flows at higher Reynolds numbers. Coupled systems of evolution equations for the interface height and fluid flux can be derived as an alternative, notably the integral boundary layer formulation of Kapitza \cite{kapitza1948wave} and Shkadov \cite{shkadov1967wave}, and the weighted residual models \cite{ruyer1998modeling,ruyer2000improved,scheid2006wave}.

The lowest rung in the hierarchy of models is occupied by the weakly nonlinear evolution equations, and for the thin film flow problem under consideration these models are relatives of the Kuramoto--Sivashinsky equation (KSE); these models govern the dynamics of small perturbations to the flat interface solution. The classical KSE in one space dimension is written
\begin{equation}\label{KSEintro}\eta_t + \eta\eta_x + \eta_{xx} + \eta_{xxxx} = 0,\end{equation}
and is usually studied with periodic boundary conditions on the interval $[0,L]$. This partial differential equation (PDE) exhibits a range of dynamical behaviours in windows of the bifurcation parameter $L$, including steady and travelling waves, time-periodic and quasi-periodic solutions, and full spatiotemporal chaos \cite{kevrekidis1990back,depapa1,depapa2}. It has been observed that \eqref{KSEintro} possesses a finite-dimensional global attractor \cite{goodm1}, and numerical results suggest a finite energy density in the limit as the periodicity $L$ becomes large; the stronger result that the $L^{\infty}$-norm of solutions at large times is bounded uniformly as $L\rightarrow \infty$ is seen (this was proved for the steady problem \cite{Michelson1}). A number of authors have considered this question analytically \cite{NST1,goodm1,CEES1, jolly2000, bronski, otto1,otto2, ottojosien}, yet the best result is far from agreeing with the optimal numerical bound. An equipartition of energy was observed for solutions \cite{pompelce, witten, Tohequip}, with a flat range in the power spectrum for the long wavelength modes, rising to a peak at the most active mode, and then decaying exponentially for the high frequencies due to strong dissipation on small scales. Additionally, Collet et al. \cite{collet1993analyticity} showed that solutions become instantly analytic, as is observed numerically in the decay of the Fourier coefficients. A KSE may be obtained in the thin film flow context from Benney models by seeking the evolution of a small perturbation to the flat film solution. Weakly nonlinear models have been derived for the three-dimensional (3D) film problem under consideration yielding evolution equations in two spatial dimensions \cite{Nepo1, Nepo2, topper1978approximate,PhysRevE.60.4143} similar to \eqref{KSEintro} -- such an equation is the main focus of this work. Variants of \eqref{KSEintro} arise in core-annular flows \cite{papageorgiou1990nonlinear} or fibre coating problems \cite{duprat2009liquid}.

Modelling thin liquid films with the 2D Navier--Stokes equations yields evolution equations that overlook transverse and mixed mode phenomena which in some situations dominate the interface dynamics. For the overlying film problem, an initially 2D flow is seen to transition to 3D waves in both experiments and numerical simulations \cite{liu1995three,park2003three,kofman2014three,kharlamov2015transition}. Hanging film flows (often referred to as film flows on inverted substrates) are less well understood, with few experimental studies. In this case, the linear theory from the 3D formulation predicts a Rayleigh--Taylor (RT) instability in the transverse modes, independent of Reynolds number. This instability leads to the formation of rivulet structures which were observed in experiments by Charogiannis and Markides \cite{markidesexperimental}. They also observed complicated pulse dynamics in the streamwise direction on the crests of the rivulets, although their parameters were not so extreme that dripping took place. In \cite{Brun1}, the authors considered Stokes flow on an inverted substrate, deriving a Benney equation for the streamwise interface evolution (no transverse dynamics) and performed experiments. They found that the fluid parameters for which dripping occurred coincided with the parameters for which their Benney model exhibited absolute instability where small perturbations grow locally (as opposed to convective instability where perturbations are convected with the flow). It can be surmised from these results that the dripping of a hanging film can be broken down into two instabilities. Firstly, the rivulet structures form due to a transverse RT instability. Then, depending on the absolute or convective nature of the streamwise instability (which is well modelled by a 1D interface equation since the rivulets are thin in the spanwise direction), dripping may occur. In a related study, Lin and Kondic \cite{lin2012thin} considered the case of a non-wetted substrate with numerical simulations of a Benney equation, and observed that fluid fronts were unstable to a transverse fingering instability. Thin rivulets form with approximately equal width in the spanwise direction, with fast moving ``drop-like" waves appearing on the rivulets as observed in the wetted case. However, the authors do not attribute the fingering instability to be of the RT-type. The controls considered in the current work attempt to prevent dripping of liquid films by averting the initial rivulet formation at the weakly nonlinear level.

There are a number of ways to influence the interfacial dynamics of thin liquid films. Spatially varying topography may be utilised to create patterned steady states, but does not have a considerable effect on the stability of solutions
\cite{pozrikidis_1988,gaskell2004gravity,tseluiko_blyth_papageorgiou_2013,doi:10.1063/1.3211289,doi:10.1063/1.4790434,doi:10.1063/1.4917026}. The open-loop controls used in such studies are steady since the topography is fixed. The effect of an electric field which is normal or parallel to the substrate has also been studied extensively -- the former has a destabilising effect on the flow \cite{tseluiko2006wave, tomlin_papageorgiou_pavliotis_2017}. The film flow problem has also been considered with the addition of heating at the wall surface which is not necessarily uniform \cite{scheid2002nonlinear,kalliadasis2003marangoni,skotheim2003instability,kalliadasis2003thermocapillary}, and even in conjunction with substrate topography \cite{Blyth4067}. Other possibilities for influencing the dynamics of thin films include the introduction of magnetic fields \cite{PhysRevE.88.023028}, surfactants \cite{blyth_pozrikidis_2004}, and substrate microstructure or coatings to induce effective slip \cite{KALPATHY2013212}. In the current work, we focus on the control of the film interface using same-fluid blowing and suction controls at the substrate surface. For the 2D simplification of the film flow problem, the weakly nonlinear evolution of the thin film with blowing and suction controls can be modelled by the 1D KSE \eqref{KSEintro} with the introduction of some non-zero right hand side $\zeta(x,t)$ (this is not the case for the majority of the above control methodologies). The optimal control of this equation (in fact a generalisation with dispersive and electric field effects) was considered by Gomes et al. \cite{gomes2016stabilizing}. For their numerical experiments, the authors consider the optimisation of point-actuator locations (given an initial condition) to stabilise a unstable travelling wave solutions. This is in contrast to the current work, where our controls are more regular in space and actuation is not restricted to a finite set of points. An alternative to same-fluid controls at the substrate surface is air-blowing and suction controls via actuators on a plate which is parallel to the substrate on which the fluid lies \cite{park2016experimental}.

The main objective of the current paper is to present the theory and numerical experiments for the optimal control of the KSE in two space dimensions,
\begin{equation} \label{controlled2dksintro} 
\eta_t + \eta \eta_x + (1- \kappa) \eta_{xx}  -  \kappa  \eta_{yy} + \Delta^2 \eta  =  \zeta.
\end{equation}
Here $\eta$ is a zero-average perturbation of the flat film solution, and $\kappa > 0$ $(<0)$ means that the film is overlying (hanging). The inertial and gravitational effects are manifest in the second derivative terms, and the bi-Laplacian term containing mixed derivatives corresponds to surface tension effects. The function $\zeta$ is a control which corresponds to blowing and suction through the substrate surface. In this paper, we consider only passive open-loop control systems, research on active control (closed-loop feedback control) problems for \eqref{controlled2dksintro} is underway by the authors. The current paper is organised as follows: In section \ref{PhysicalModel}, we present the physical problem, and discuss the hierarchy of models culminating with \eqref{controlled2dksintro}. In section \ref{opttranssubsec1}, we consider purely transverse controls, $\zeta = \tilde{\zeta}(y,t)$. Seeking an optimal control of this form reduces to a linear quadratic regulator (LQR) problem for each transverse Fourier mode. The resulting optimal control is used to suppress the exponential growth of the purely transverse modes for a hanging film set-up. In section \ref{FOCref}, we discuss both the analytical and numerical aspects of the full optimal control problem, where $\zeta \equiv \zeta(x,y,t)$. Existence of an optimal control is proven, and a forward--backward sweeping method is employed for numerical experiments. Finally, in section \ref{TranwaveeffectsSec}, we study the impact of transverse modes in the dynamics of the streamwise flow evolution. This study cannot be categorised as an optimal control problem, however, our findings are linked to results in section \ref{FOCref}, and may lead to new control strategies for flows with a dominant direction. Our conclusions and a discussion are presented in section \ref{CONC1}.

\section{Physical problem and hierarchy of models\label{PhysicalModel}}

\subsection{Physical problem and governing equations}

\begin{figure}
\begin{tikzpicture}[scale=1.40]

    \draw (5,0) coordinate (A1) -- (1,3) coordinate (A2);
    \draw (9,0.5) coordinate (A3) -- (5,3.5) coordinate (A4);
    \draw (A1) -- (A3); \draw (A2) -- (A4);
    \fill[gray,opacity=0.30] (A1) -- (A2) -- (A4) -- (A3) -- cycle;
    
             \def\mypath{ (A1) -- (A2) -- (A4) -- (A3) -- (A1) }

%
%
%
%
%
            

    \draw[dashed] (A1) -- (4,0.35) node[above] {$\theta$};
    \draw (4.25,0.55) arc (135:175:0.4);

  \draw[-] (-0.45, 4.1) node[right] {$z=h(x,y,t)$};
    
  \draw[-] (0.2, 3) node[right] {$z=0$};

    \draw[dashed, opacity=0.0] (5.2,1.466-0.35) coordinate (B1) -- (1.2,4.466-0.35) coordinate (B2);
    \draw[dashed, opacity=0.0] (9.2,1.966-0.35) coordinate (B3) -- (5.2,4.966-0.35) coordinate (B4);
    \draw[dashed, opacity=0.0] (B1) -- (B3); 
    \draw[dashed, opacity=0.0] (B2) -- (B4);

    \draw[->] (9.2, 2) -- (9.7,2.0625) node[right] {$y$};
    \draw[->] (9.2, 2) -- (9.7, 1.625) node[right] {$x$};
    \draw[->] (9.2, 2) -- (9.4, 2.466) node[above] {$z$};
    
     \draw[->] (8.4, 3.7) -- (8.4, 3.1);
     
      \draw[-] (8.4, 3.4) node[right] {$\bm{g}$};

   \filldraw[gray,opacity=0.3, very thick]                    (B1)
  decorate [decoration={snake,amplitude=0.12cm,segment length=0.911cm}]             { -- (B2) }
  decorate [decoration={snake,amplitude=0.04cm,segment length=1.5cm}] 		{ -- (B4) }
  decorate [decoration={snake,amplitude=0.12cm,segment length=0.911cm}]           { -- (B3) }
   decorate [decoration={snake,amplitude=0.04cm,segment length=1.5cm}]           { -- (B1) }
      ;
    
    \pgftransformcm{0.2857}{0.0357}{0.2857}{-0.2143}{\pgfpoint{20}{0}}
    \foreach \x in {0,1,...,27}{
      \foreach \y in {0,1,...,27}{
        \node[draw,circle,inner sep=0.001pt,fill] at (\x*0.5+13.20,\y*0.5-11.65) {};
      }
    }

\end{tikzpicture}\caption{Schematic of the problem.} \label{Setupdiagram}
\end{figure}
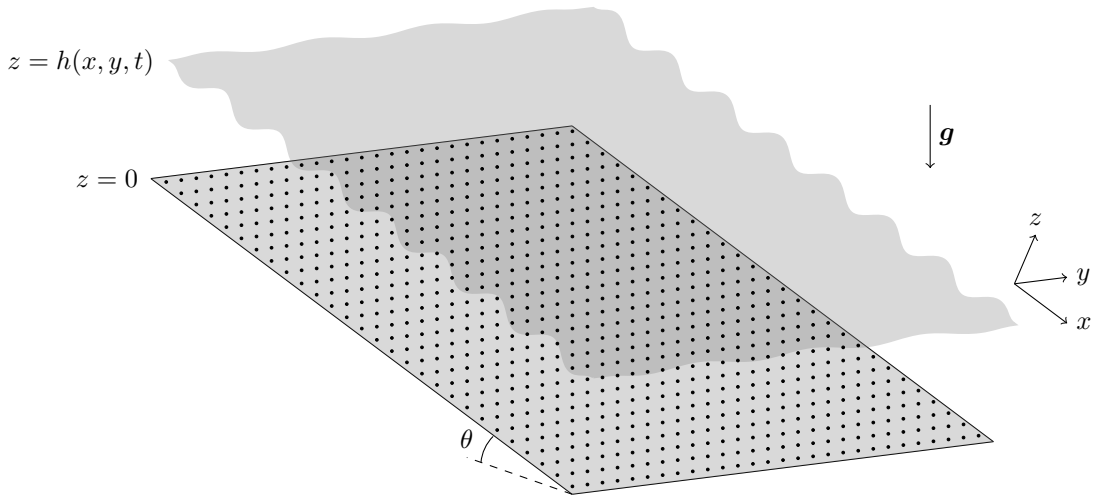

We consider a Newtonian fluid with constant density $\rho$, dynamic viscosity $\mu$, and kinematic viscosity $\nu$, flowing under gravity along a flat infinite 2D substrate inclined at a non-zero angle $\theta$ to the horizontal -- the schematic of the problem is presented in Figure \ref{Setupdiagram}. Same-fluid blowing and suction controls are applied at the substrate surface; the forcing is assumed to be purely perpendicular to the substrate and is introduced into the problem through modifying the impermeability condition at the substrate, but not modifying the no-slip condition. The actuator locations are assumed to be close together so that we may assume that blowing or suction can be performed at any point on the substrate surface (the corresponding study of point-actuated control by the authors will be presented elsewhere). Initially, this presents itself as a boundary control problem, however it manifests itself as a distributed control problem for the dynamics of the fluid--air interface as we will show. We use coordinates $(x,y,z)$ which are fixed in the plane, with $x$ directed in the streamwise direction, $y$ in the spanwise direction, and $z$ perpendicular to the substrate, as shown for the case of an overlying film in the schematic of Figure \ref{Setupdiagram}. Note that, as $\theta$ increases, the substrate and axes rotate; we have $\theta \in (0, \pi/2)$ for overlying films, $\theta \in (\pi/2,\pi)$ for hanging films, and the vertical film flow with $\theta = \pi/2$. The surface tension coefficient between the liquid and the surrounding hydrodynamically passive medium is denoted by $\sigma$ (assumed constant), the local film thickness is denoted by $h(x,y,t)$, a function of space and time, with unperturbed thickness $\ell$, and the acceleration due to gravity is denoted by $\bm{g} = (g\sin\theta,0,-g\cos\theta)$. In the case when no blowing or suction is applied, the Navier--Stokes equations (with no-slip and impermeability boundary conditions at the substrate, and the kinematic condition and stress balances at the interface) admit an exact Nusselt solution \cite{nusselt1,FLM:367246} with a film of uniform thickness, i.e. $h=\ell$, and a streamwise velocity profile which is semi-parabolic in $z$. Velocities and controls are rescaled with the base velocity of this exact solution at the free surface, $U_0 = g\ell^2\sin\theta/2\nu$, the lengths scale with $\ell$, and we take the viscous scaling for the pressure. We use the non-dimensional parameters
\begin{equation}\label{dimensionlessparameters} {\Rey} = \frac{U_0 \ell}{\nu} = \frac{g \ell^3\sin\theta}{2\nu^2}, \qquad {\Capil} = \frac{U_0 \mu}{\sigma} = \frac{ \rho g \ell^2\sin\theta}{2\sigma},\end{equation}
where the Reynolds number ${\Rey}$ measures the ratio of inertial to viscous forces, and the capillary number ${\Capil}$ measures the ratio of surface tension to viscous forces. The non-dimensional Navier--Stokes equations for this problem are
\begin{subequations}
\begin{align}{\Rey} \left( \bm{u}_t + (\bm{u} \bm{\cdot} \bm{\nabla} )\bm{u} \right)= & - \bm{\nabla} p  + {\nabla}^2 \bm{u} +  2 \bm{\overline{g}} ,\\
\bm{\nabla}\bm{\cdot}\bm{u} = & \; 0,
\end{align}
\end{subequations}
where $\bm{u}= (u,v,w)$ is the velocity field, $p$ is the pressure, $\bm{\nabla}$ is the 3D spatial gradient operator, ${\nabla}^2 = \bm{\nabla} \bm{\cdot} \bm{\nabla}$ is the 3D Laplacian operator, and $\bm{\overline{g}} = (1,0,-\cot\theta)$ is the non-dimensional gravitational forcing. We have no-slip conditions at the solid substrate surface, $u|_{z=0} = v|_{z=0}= 0$, and the impermeability condition is modified to account for the controls as $w|_{z=0} = f(x,y,t)$. The non-dimensional kinematic condition and balance of stresses in the two tangential directions and normal direction at the interface are omitted for brevity, but are given (with an additional electric field term) in \cite{tomlin_papageorgiou_pavliotis_2017}.

Actuation through flow boundaries (as is considered in this work) for the Navier--Stokes equations has been considered in the optimal control of turbulent channel flows by Bewley and Moin \cite{bewley1994optimal}, and in 2D with wall slip by Chemetov and Cipriano \cite{doi:10.1137/17M1121196}. Existence and uniqueness results for optimal controls in the case of body forcing controls (where $f$ appears on the right hand side of the Navier--Stokes equations) are given in \cite{BEWLEY20001007,Abergel1990}. Additionally, optimisation of wall-parallel velocities for cavity driven flows have been considered, as well as other problems related to optimisation of fluid mixing or thermal convection.

\subsection{Fully nonlinear Benney equation} 
In order to reduce the Navier--Stokes formulation to a forced evolution equation for the film thickness, we make a long-wave assumption. 
The details of the following derivation are omitted since they are similar to those in \cite{tomlin_papageorgiou_pavliotis_2017} for the case of a fully 3D model of an electrified thin film, and the derivation for the 2D case is provided by \cite{thompson2016falling}. We assume that the typical interfacial deformation wavelengths $\lambda$ are large compared to the unperturbed thickness $\ell$, set $\delta = \ell /\lambda \ll 1$, and introduce the change of variables
\begin{equation}\label{longwavereg1changeofvars} 
x = \frac{1}{\delta} \widehat{x},\qquad y = \frac{1}{\delta} \widehat{y}, \qquad t = \frac{1}{\delta} \widehat{t},\qquad w = \delta \widehat{w},\qquad f = \delta \widehat{f}, \qquad {\Capil} = \delta^2 \overline{{\Capil}},
\end{equation}
where hats denote $O(1)$ quantities, and the capillary number is rescaled in order to retain the effects of surface tension in the leading order dynamics, with $\overline{\Capil} = O(1)$. We also assume that the Reynolds number ${\Rey}$ is an $O(1)$ quantity. The change of variables (\ref{longwavereg1changeofvars}) is substituted into the governing equations and the hats are dropped. Then, with a systematic asymptotics procedure, in which the flow field is substituted into the kinematic equation, we obtain a fully nonlinear Benney equation in two spatial dimensions, with errors of $O(\delta^2)$,
\begin{equation} \label{benneyqequation1}  H_t + \bm{\nabla} \bm{\cdot} \left[ \left( \frac{2}{3} H^3 + \frac{8{\Rey}}{15} \delta  H^6 H_x - \frac{2{\Rey}}{3} \delta H^4 F \right) \bm{e}_x  -  \frac{2}{3} \delta H^3 \bm{\nabla} \left( H \cot\theta - \frac{1}{2\overline{{\Capil}}} \Delta H \right) \right] = F,  \end{equation}
where we have redefined $\bm{\nabla} = (\partial_x, \partial_y)$, $\bm{e}_x = (1,0)$, and $\Delta\equiv \partial_x^2+\partial_y^2$ is the 2D Laplacian operator. The variables $H$ and $F$ are approximations of the interface height $h$ and given control $f$, respectively, correct to $O(\delta)$. Notice that $F$ is not only on the right hand side of \eqref{benneyqequation1}, but appears in a nonlinear term in the streamwise direction at $O(\delta)$. We have assumed that $F_t = O(1)$ in the current variables, otherwise a term involving $F_t$ is promoted from the $O(\delta^2)$ error to the $O(\delta)$ terms of \eqref{benneyqequation1}. As observed by a number of authors \cite{pumir1983solitary,joo1991long,rosenau1992bounded}, the uncontrolled 1D simplification (no spanwise variation) of this Benney equation exhibits finite-time blow-ups in numerical simulations. The existence of an optimal control for \eqref{benneyqequation1} is an open problem (lack of analytical results for the uncontrolled equation is problematic), and we are not aware of any numerical studies of optimal control for the 1D case or related thin film equations. However, feedback control methods for the 1D version of \eqref{benneyqequation1} have been considered by Thompson et al. \cite{thompson2016stabilising}.

\subsection{Weakly nonlinear evolution: 2D Kuramoto--Sivashinsky equation\label{sec:KSE}}

We seek to analyse the evolution of a sufficiently small perturbation about the non-dimensional exact constant solution to \eqref{benneyqequation1} given by $H=1$, $F=0$. For this, we substitute $H = 1 + \delta \eta$ and $F = 4 \delta^2  \zeta$ into equation \eqref{benneyqequation1} where $\eta, \zeta$ are $O(1)$ quantities, and also assume that $\cot\theta$ is $O(1)$. Truncating terms of $o(\delta)$, the resulting equation is
\begin{equation} \eta_t +  2 \eta_x + 4 \delta \eta \eta_x + \frac{8 {\Rey}}{15} \delta \eta_{xx}  -  \frac{2}{3} \delta  \eta_{xx} \cot\theta -  \frac{2}{3} \delta  \eta_{yy} \cot\theta +  \frac{1}{3\overline{{\Capil}}} \delta \Delta^2 \eta   =  4 \delta \zeta .\end{equation}
Rescaling with
\begin{equation}\left\{ \label{KSEchangeofvars1}
\begin{array}{c} {\displaystyle  t = \frac{75}{64\delta\overline{{\Capil}}{\Rey}^2} \; \widehat{t} , \qquad x -2t = \frac{\sqrt{5}}{2\sqrt{2} \; \overline{{\Capil}}^{1/2} {\Rey}^{1/2}} \; \widehat{x}, \qquad y =  \frac{\sqrt{5}}{2\sqrt{2} \; \overline{{\Capil}}^{1/2} {\Rey}^{1/2}} \; \widehat{y},} \\[8pt]
{\displaystyle \eta = \frac{4\sqrt{2}\; \overline{{\Capil}}^{1/2} {\Rey}^{3/2} }{15\sqrt{5}} \; \widehat{\eta} , \qquad \zeta =  \frac{64 \sqrt{2} \; \overline{{\Capil}}^{3/2} {\Rey}^{7/2} }{1125 \sqrt{5}} \; \widehat{\zeta}, }
\end{array}   \right.
\end{equation}
and dropping hats gives
\begin{equation} \label{controlled2dks} 
\eta_t + \eta \eta_x + (1- \kappa) \eta_{xx}  -  \kappa  \eta_{yy} + \Delta^2 \eta  =  \zeta,
\end{equation}
where $\kappa = 5 \cot\theta/4{\Rey}$. Note that the rescaling \eqref{KSEchangeofvars1} involves a Galilean transformation, so \eqref{controlled2dks} is in a moving reference frame. For our investigation of optimal controls with actuators at every point on the substrate, the advection term makes no difference and the result in the moving frame can be translated back into the ``lab" frame without any issue. Similarly, the advection in the streamwise direction has no effect on the purely spanwise forcing we study in section \ref{TranwaveeffectsSec}. The case of point-actuated controls requires consideration of travelling actuator grids, however, and the study of this situation is presented elsewhere.

We supplement \eqref{controlled2dks} with periodic boundary conditions on the rectangle $Q = [0,L_1]\times[0,L_2]$. For this purpose, we denote the wavenumber vectors by $\bm{\tilde{k}}$, with components
\begin{equation}\label{ktildedefn}\tilde{k}_1 = \frac{2\pi k_1}{L_1},\qquad \tilde{k}_2 = \frac{2\pi k_2}{L_2},\end{equation}
for $\bm{k} \in \mathbb{Z}^2$, so that $\eta$ and $\zeta$ may be written in terms of their Fourier series as
\begin{equation}\label{fourierseriesetazeta1}\eta = \sum_{\bm{k}\in\mathbb{Z}^2} \eta_{\bm{k}} e^{i \bm{\tilde{k}} \bm{\cdot} \bm{x}}, \qquad \zeta = \sum_{\bm{k}\in\mathbb{Z}^2} \zeta_{\bm{k}} e^{i \bm{\tilde{k}} \bm{\cdot} \bm{x}}.\end{equation}
Here $\eta_{-\bm{k}}$ and $\zeta_{-\bm{k}}$ are the complex conjugates of $\eta_{\bm{k}}$ and $\zeta_{\bm{k}}$, respectively, since both the solution and control are real-valued. Equation \eqref{controlled2dks} is equivalent to the infinite-dimensional system of ODEs for the Fourier coefficients,
\begin{equation}\label{fouriercoeffsfull1}\frac{\mathrm{d}}{\mathrm{d}t}\eta_{\bm{k}} + \frac{i\tilde{k}_1}{2} \sum_{\bm{m}\in\mathbb{Z}^2} \eta_{\bm{k}-\bm{m}} \eta_{\bm{m}} = \left[ (1-\kappa) \tilde{k}_1^2 - \kappa \tilde{k}_2^2 - |\bm{\tilde{k}}|^4  \right] \eta_{\bm{k}} + \zeta_{\bm{k}},\end{equation}
for each $\bm{k} \in \mathbb{Z}^2$. Since \eqref{controlled2dks} governs the perturbation of the flat film state, we consider initial conditions, denoted by $v$, with zero spatial mean ($v_{\bm{0}} = 0$). From \eqref{fouriercoeffsfull1} with $\bm{k} = \bm{0}$, the evolution of the mean depends upon the control $\zeta$ as
\begin{equation}\frac{\mathrm{d}}{\mathrm{d}t} \eta_{\bm{0}} = \zeta_{\bm{0}}.\end{equation}
We restrict to controls with $\zeta_{\bm{0}} \equiv 0$, so that the zero spatial mean of the solution is preserved.
The controls considered in this work can thus be thought of as moving fluid from one location to another, without changing the total fluid volume.

From the ODE description \eqref{fouriercoeffsfull1}, it can be observed that the dynamics of the transverse modes (i.e. $k_1 = 0$) are linear, being governed by
\begin{equation}\label{Fouriercoefftranssystem1earlier}\frac{\mathrm{d}}{\mathrm{d}t} \tilde{\eta}_{k_2}  =  (-\kappa \tilde{k}_2^2 - \tilde{k}_2^4) \tilde{\eta}_{k_2}+  \tilde{\zeta}_{k_2},\end{equation}
where we have denoted the transverse Fourier coefficients $\eta_{(0,k_2)}$ by $\tilde{\eta}_{k_2}$, with the same notation for the transverse components of the forcing $\zeta$. We let $\tilde{P}$ denote the projection onto the subspace of transverse modes, and define $\tilde{\eta}(y,t)$ to be the image of $\eta$ under this projection, $\tilde{\eta} = \tilde{P} \eta$. This projection satisfies the linear PDE
\begin{equation}\label{transversesystem1}\tilde{\eta}_t - \kappa \tilde{\eta}_{yy} + \tilde{\eta}_{yyyy} = \tilde{P}\zeta,\end{equation}
which is the equivalent of the ODE system \eqref{Fouriercoefftranssystem1earlier}, and may be obtained from 
\eqref{controlled2dks} by averaging over the streamwise direction. From \eqref{fouriercoeffsfull1}, it can be seen that the dynamics of the streamwise and mixed modes are slaved to $\tilde{\eta}$, i.e. the transverse modes only decouple partially from the full nonlinear system.

The parameter $\kappa$ encodes the incline of the substrate, and in the absence of control, we have three distinct dynamical regimes depending on its value. For $\kappa < 0$ we obtain hanging films, and, for $\zeta = 0$, a range of transverse modes with $0 < |\tilde{k}_2| <  (- \kappa)^{1/2}$ are linearly unstable. Since the governing equation \eqref{Fouriercoefftranssystem1earlier} is linear, unbounded exponential growth occurs with rate $- \kappa \tilde{k}_2^2 - \tilde{k}_2^4 > 0$. A large focus of this work is on stabilising hanging films with growing transverse modes. For $\kappa \geq 1$ the flow is overlying and the Reynolds number is subcritical, with $\kappa = 1$ corresponding to the critical Reynolds number ${\Rey}_c = 5\cot\theta/4$. In this case, simple energy estimates may be used to show that all solutions converge to zero in the absence of blowing or suction; importantly, the quadratic nonlinearity facilitates the transfer of energy, rather than being a sink or source of energy. Overlying flows with supercritical Reynolds numbers are then found for $0 \leq \kappa < 1$. When $\kappa = 0$, the film is vertical, and the canonical equation \eqref{controlled2dks} reduces to a forced version of the thin film equation obtained by Nepomnyashchy \cite{Nepo1,Nepo2},
\begin{equation} 
\label{nepoequation1231}\eta_t + \eta \eta_x + \eta_{xx}  + \Delta^2 \eta  =  \zeta.
\end{equation}
Without controls, this equation has been studied both analytically and numerically by a number of authors. Pinto \cite{pinto1,pinto2} proved the existence of a finite-dimensional global attractor, and instant analyticity of solutions. Numerical studies have been carried out by Akrivis et al.~\cite{akrivislinearly} and Tomlin et al.~\cite{tkp2018}, where it was shown that solutions possess a finite energy density independent of the periodic domain size with chaotic dynamics emerging on sufficiently large domains. Similar Kuramoto--Sivashinsky-type dynamics are found for $0<\kappa <1$.

We apply controls over a finite time interval $[0,T]$, in order to drive the solution to a zero-mean desired state $\overline{\eta}(\bm{x},t)$. The desired state is not required to be an exact (possibly unstable) solution of the uncontrolled problem, although many control methodologies take this assumption. For our setting, we employ homogeneous spatial norms, with all functions assumed to be $Q$-periodic and have zero spatial mean for all times; the spaces $H_0^s$ (where $H_0^0 = L_0^2$), $L^2(0,T;H_0^s)$, and $C^0([0,T]; H_0^s)$, are defined through their respective norms:
\begin{subeqnarray}\label{Hsnorm1aaa}
\gdef\thesubequation{\theequation \textit{a,b}}
\qquad\| \eta \|_{H_0^s}^2 =  \sum_{ \bm{k} \in\mathbb{Z}^2\backslash\{\bm{0}\} } | \bm{\tilde{k}} |^{2s} | \eta_{\bm{k}} |^2, \qquad \| \eta \|_{L^2(0,T;H_0^s)}^2 = \int_0^T \| \eta \|_{H_0^s}^2 \; \mathrm{d}t,\\
\gdef\thesubequation{\theequation \textit{c}}
\| \eta \|_{C^0([0,T];H_0^s)} = \sup_{t\in[0,T]} \| \eta \|_{H_0^s}.\qquad\qquad\qquad\quad
\end{subeqnarray}
Note that the Fourier symbol of the operator $(-\Delta)^{s/2}$ is $|\bm{\tilde{k}} |^{s}$, so the $L_0^2$-norm of $\Delta \eta$ equals the $H_0^2$-norm of $\eta$, for example. Inner products for the $H_0^s$ and $L^2(0,T;H_0^s)$ spaces are denoted by angled brackets with the appropriate subscripts. The space of admissible controls is denoted $F_{\textrm{ad}}$, a non-empty, closed, convex subset of $L^2(0,T;L_0^2)$, where the norm is defined by (\ref{Hsnorm1aaa}b) with $s=0$. We optimise with respect to the cost functional
\begin{equation}\label{costfunctional1aa}\mathcal{C}_{s,\gamma}(\eta,\zeta;\overline{\eta}) = \frac{1}{2}  \| \eta - \overline{\eta} \|_{L^2(0,T;H_0^s)}^2  + \frac{1}{2} \| \eta(\cdot,T) - \overline{\eta}(\cdot,T) \|_{H_0^s}^2 + \frac{\gamma}{2}  \| \zeta \|_{L^2(0,T;L_0^2)}^2
\end{equation}
where $s \in \mathbb{R}$ and $\gamma > 0$ ($\gamma = 0$ allows infinitely strong controls). We denote the three components of the cost functional by $\mathcal{C}_{s,\gamma}^{(1)}$, $\mathcal{C}_{s,\gamma}^{(2)}$, and $\mathcal{C}_{s,\gamma}^{(3)}$, respectively. The optimal control $\zeta^*$ and associated optimal state $\eta^*$ (which solves the KSE \eqref{controlled2dks} with $\zeta = \zeta^*$) are defined as minimisers of the cost \eqref{costfunctional1aa} over all controls in $F_{\textrm{ad}}$, i.e. if $\zeta' \in F_{\textrm{ad}}$ has associated state $\eta'$, then $\mathcal{C}_{s,\gamma}(\eta^*,\zeta^*;\overline{\eta}) \leq \mathcal{C}_{s,\gamma}(\eta',\zeta';\overline{\eta})$. The parameter $\gamma$ in \eqref{costfunctional1aa} can be thought of as the cost of using the control, relative to the cost of inaccuracy between $\eta$ and the desired state $\overline{\eta}$. For small $\gamma$ we may use larger controls to ensure that the solution is very close to the desired state, but for large $\gamma$ the controls are expensive and the difference between $\eta$ and $\overline{\eta}$ is less important. Larger values of the Sobolev index $s$ have the effect of increasing the weighting on the higher frequencies with larger wavelength modes in the solution costing relatively less. The second term in the cost functional \eqref{costfunctional1aa} ensures that the solution cannot become unbounded in $H_0^s$ at the final time.

The definition of the spatial norm used here is a measure of energy density; norm (\ref{Hsnorm1aaa}a) is related to the usual integral definition for the square of the Sobolev norm through multiplication by $L_1L_2$. Thus, the cost functionals for different choices of the domain $Q$ become comparable as they are independent of the underlying periodicities; this makes sense if we think of our problem as being on an infinite domain on which we impose periodicity, rather than on a bounded domain with periodic boundary conditions. Taking this definition is natural for the study of \eqref{controlled2dks}, since it has been observed that \eqref{nepoequation1231} without forcing ($\zeta = 0$) possesses a finite energy density as the domain size is increased \cite{tkp2018}. Furthermore, values of the norms are comparable to the amplitude of the interface.

For the numerical study of \eqref{controlled2dks} on $Q$-periodic domains, we utilise a family of implicit--explicit backwards differentiation formula (BDF) methods constructed by Akrivis et al. \cite{akrivis2004linearly} for a class of nonlinear parabolic equations under appropriate assumptions on the linear and nonlinear terms. They considered evolution equations of the form 
\begin{equation}\eta_t + \mathcal{A}\eta = \mathcal{B}(\eta),\end{equation}
where $\mathcal{A}$ is a positive definite, self-adjoint linear operator, and $\mathcal{B}$ is a nonlinear operator which satisfies a local Lipschitz condition. It was shown that these numerical schemes are efficient, convergent, and unconditionally stable. For us, we have the addition of forcing $\zeta$ to the right hand side, and thus the operators are defined as
\begin{subeqnarray}
\label{AandBdefns1}  
\gdef\thesubequation{\theequation \text{a,b}}
\mathcal{A}\eta = (1- \kappa) \eta_{xx}  -  \kappa  \eta_{yy} + \Delta^2 \eta + c\eta , \qquad \mathcal{B}(\eta,\zeta)  = - \eta \eta_x + \zeta +  c\eta,
\end{subeqnarray}
where $c$ is chosen to ensure that $\mathcal{A}$ is positive definite. These schemes were utilised for the corresponding 1D optimal control problem by Gomes et al. \cite{gomes2016stabilizing}, and their applicability was checked for similar multi-dimensional problems in \cite{akrivislinearly} (including a convergence study) and \cite{tomlin_papageorgiou_pavliotis_2017} for a non-local problem. In order to perform computations, we truncate the Fourier series \eqref{fourierseriesetazeta1} to $|k_1| \leq M$, $|k_2|\leq N$, which corresponds to a discretisation of the spatial domain $Q$ into $(2M + 1)\times (2N+1)$ equidistant points, and carry out time-integration of the system in Fourier space using the BDF methods.

\section{Optimal transverse control for hanging films\label{opttranssubsec1}}

We first consider controlling the transverse instabilities present for hanging films ($\kappa < 0$) by applying controls of the form $\zeta = \tilde{\zeta}(y,t)$ (with $\tilde{\zeta}_{0} = 0$). We work under the assumption that the dynamics of 
\eqref{controlled2dks} are bounded if the linear growth in the spanwise dimension is controlled. This is a reasonable assumption given the form of the ODE system \eqref{fouriercoeffsfull1}\footnote{Either side of $\kappa = 0$, there are unstable streamwise and mixed modes, governed by \eqref{fouriercoeffsfull1}; decreasing $\kappa$ only increases the strength of the destabilising terms relative to the nonlinearity and fourth order damping (and increases the number of unstable modes, similar to the effect of increasing $L_1$ and $L_2$), rather than qualitatively changing the structure of the nonlinear system. The critical value of $\kappa = 0$ is only important for the transverse system \eqref{transversesystem1}, as it is the point at which unstable modes first appear.}, and is confirmed by our numerical results. The ODE system \eqref{Fouriercoefftranssystem1earlier} has the explicit solution
\begin{equation}\label{explicitsolutiontrans1}\tilde{\eta}_{k_2}(t) = e^{-(\kappa \tilde{k}_2^2 + \tilde{k}_2^4)t }\tilde{v}_{k_2}  + e^{-(\kappa \tilde{k}_2^2 + \tilde{k}_2^4)t } \int_0^t \tilde{\zeta}_{k_2}(\tau) e^{(\kappa \tilde{k}_2^2 + \tilde{k}_2^4)\tau } \; \mathrm{d}\tau,\end{equation}
where $\tilde{v}_{k_2}$ are the transverse Fourier coefficients of a given initial condition $v$. Equation \eqref{explicitsolutiontrans1} is often referred to as the control-to-state map. Let $\tilde{P}_{\Xi}$ denote the projection onto the subspace spanned by the unstable purely transverse modes, where the wavenumbers belong to the set $\Xi = \{ k_2 \in \mathbb{Z} \; | \; 0 < |\tilde{k}_2| <  (- \kappa)^{1/2} \},$
and define $\tilde{\eta}_{\textrm{u}} = \tilde{P}_{\Xi} \eta$. The transverse modes in $\mathbb{Z}\backslash \Xi$ are either neutrally stable or decay exponentially without controls, and have no effect on the boundedness of the solution. Thus, in this section, we take $F_{\textrm{ad}}$ to be the image of $L^2(0,T;L_0^2)$ under the operator $\tilde{P}_{\Xi}$. In fact, the linearity of \eqref{Fouriercoefftranssystem1earlier} permits the explicit construction of an optimal control which is smooth in both time and space. Since the other transverse modes are damped, the long time dynamics will converge to those of the full system \eqref{controlled2dks} with $\tilde{\eta} = 0$ if the modes in $\Xi$ are controlled to zero. The desired state is taken to be the orthogonal projection $\overline{\eta} = (I -\tilde{P}_{\Xi}) \eta  = \eta - \tilde{\eta}_{\textrm{u}}$, with no components of unstable transverse modes. It is important to note that due to the slaving of the streamwise and mixed modes to the transverse modes, $\overline{\eta}$ is not the same as the evolution of an initial condition $(I - \tilde{P}_{\Xi})v$ (in other words, the evolution under the PDE and the projection $I - \tilde{P}_{\Xi}$ do not commute). Note also that the desired state is dependent on the solution locally in time, and is not a pre-prescribed function. The cost functional \eqref{costfunctional1aa} simplifies to 
\begin{equation}\label{transversecostfunc111}\mathcal{C}_{s,\gamma}(\eta,\tilde{\zeta}) = \frac{1}{2} \sum_{k_2 \in\Xi} \left[\int_0^T  |\tilde{k}_2|^{2s} |\tilde{\eta}_{k_2}(t)|^2 \; \mathrm{d}t +  |\tilde{k}_2|^{2s} |\tilde{\eta}_{k_2}(T)|^2 + \gamma \int_{0}^T | \tilde{\zeta}_{k_2}(t) |^2 \; \mathrm{d}t \right],\end{equation}
where the summands are costs for each individual mode. We denote the optimal solution and control with respect to the cost functional by $\eta^*$ and $\tilde{\zeta}^*$, respectively. For finite $T$, this results in a set of linear ODE control problems, each a linear quadratic regulator (LQR) problem, which have a well developed theory (see \cite{zabczyk2009mathematical} for example). The optimal control (derived in appendix \ref{appendixdericoptlinear}) is given by
$\tilde{\zeta}_{k_2}^* = r_{k_2}\tilde{\eta}_{k_2}^*$ (with stars denoting optimality) where $r_{k_2}$ satisfies a Riccati equation and final time boundary condition
\begin{equation}\label{ricattieqn1}\frac{\mathrm{d}}{\mathrm{d}t}r_{k_2} = - r_{k_2}^2 + 2 (\kappa \tilde{k}_2^2 + \tilde{k}_2^4) r_{k_2} + \frac{|\tilde{k}_2|^{2s}}{\gamma}, \quad r_{k_2}(T) = - |\tilde{k}_2|^{2s}/\gamma.\end{equation}
By defining the roots
\begin{equation}\lambda_{\pm} =  (\kappa \tilde{k}_2^2 + \tilde{k}_2^4)  \pm \sqrt{ (\kappa \tilde{k}_2^2 + \tilde{k}_2^4)^2 + \frac{|\tilde{k}_2|^{2s}}{\gamma}},\end{equation}
we can give the solution to \eqref{ricattieqn1} explicitly as
\begin{equation}r_{k_2} = \frac{\lambda_{-}   \exp(\lambda_{+}C+ \lambda_{-}t)   - \lambda_{+}\exp(\lambda_{+}t+\lambda_{-}C)}
{ \exp(\lambda_{+}C+ \lambda_{-}t) - \exp(\lambda_{+}t+\lambda_{-}C)}, \qquad  C = \frac{1}{\lambda_{-} - \lambda_{+}} \log\left(\frac{\lambda_{-}(\lambda_{+} - 1)}{\lambda_{+}(\lambda_{-} - 1)}\right) + T.\end{equation}
Furthermore, the value of the cost functional attained by the optimal control is
\begin{equation}\label{costreduced1}\mathcal{C}_{s,\gamma}^* = - \frac{\gamma}{2} \sum_{k_2 \in \Xi} r_{k_2}(0)| \tilde{v}_{k_2} |^2.\end{equation}
Since the problem under consideration is linear and finite-dimensional, these optimal controls are unique.

For comparison, in the infinite time-horizon case, i.e. $T=\infty$, we find $\dot{r}_{k_2} = 0$ so that $r_{k_2} = \lambda_{-}$ (this is the correct root to stabilise the system \eqref{ricattieqn1}). Substituting back into (\ref{Fouriercoefftranssystem1earlier}) yields the solution
\begin{equation}\tilde{\eta}_{k_2}^* =  \exp\left( - t \; \sqrt{ (\kappa \tilde{k}_2^2 + \tilde{k}_2^4)^2 + |\tilde{k}_2|^{2s}/\gamma} \right)\tilde{v}_{k_2},\end{equation}
where the optimal control is defined by $\tilde{\zeta}_{k_2}^* = \lambda_{-} \tilde{\eta}_{k_2}^*$.

\begin{figure}
\centering
\begin{subfigure}{2.9in}
\caption{ $L_0^2$-norms for case (i).} 
\includegraphics[width=2.9in]{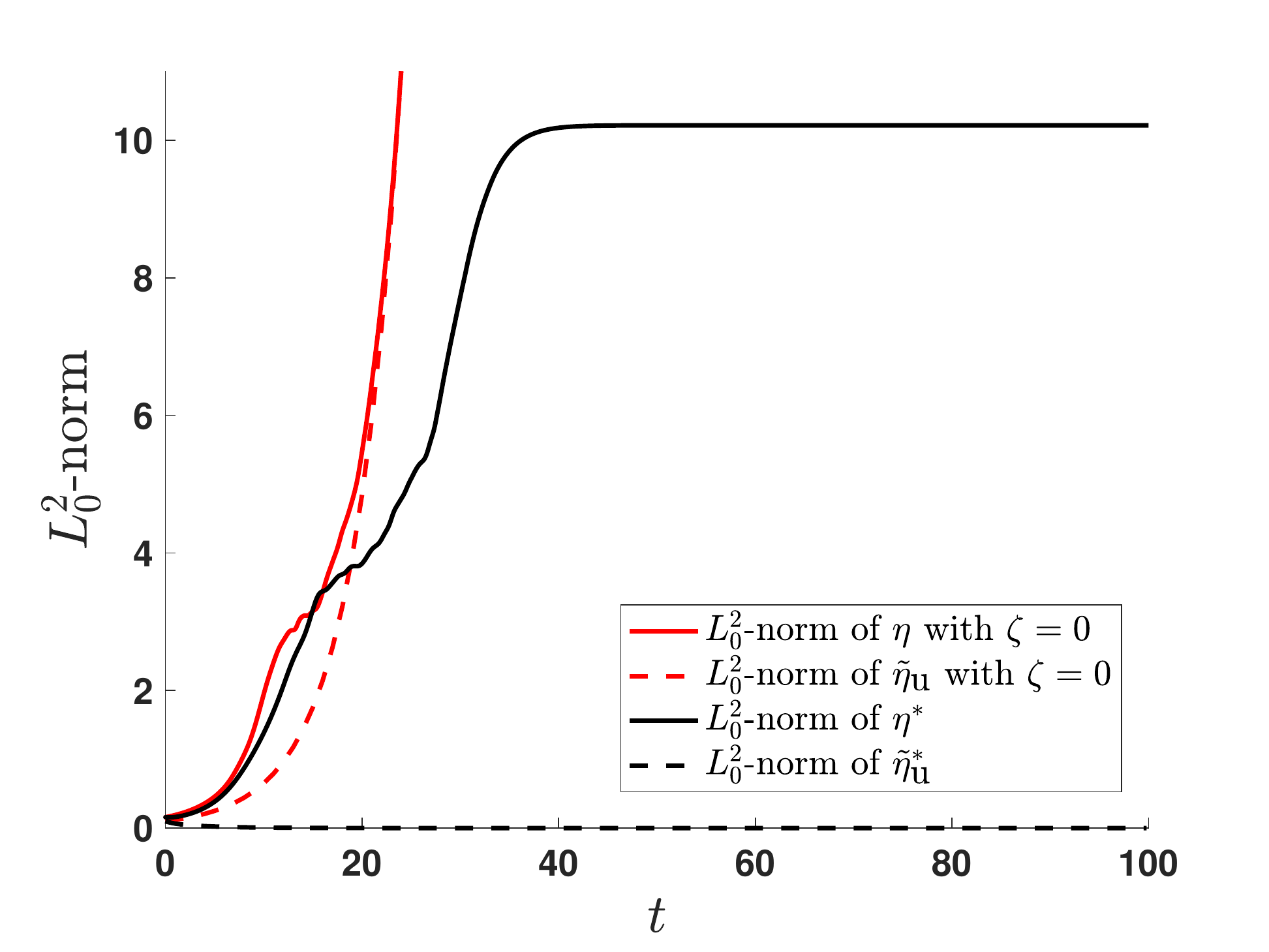}
\end{subfigure}
\begin{subfigure}{2.9in}
\caption{ $L_0^2$-norms for case (ii).} 
\includegraphics[width=2.9in]{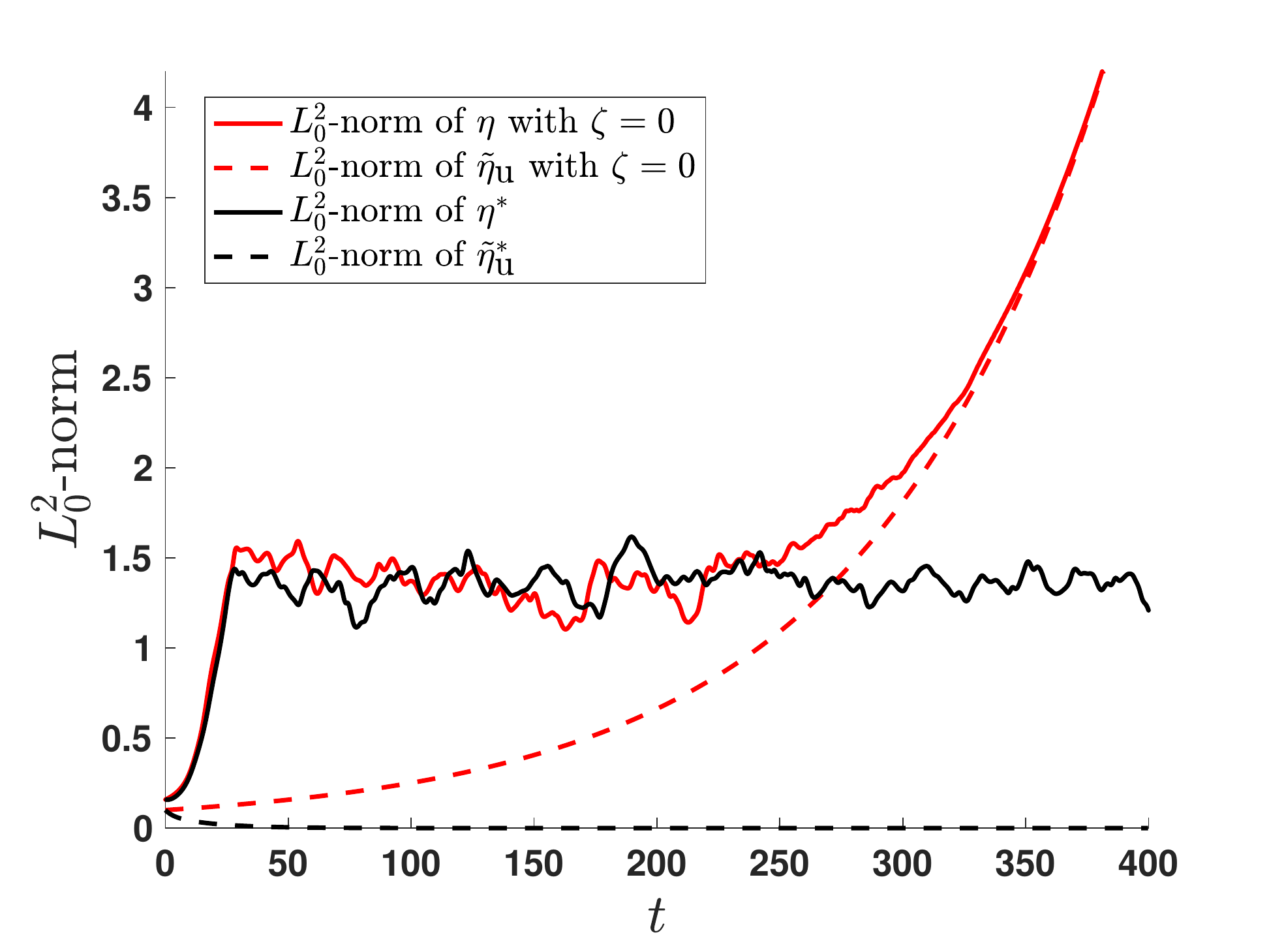}
\end{subfigure}
\begin{subfigure}{2.9in}
\caption{Components of $\mathcal{C}_{2,1}^*$ for case (i).} 
\includegraphics[width=2.9in]{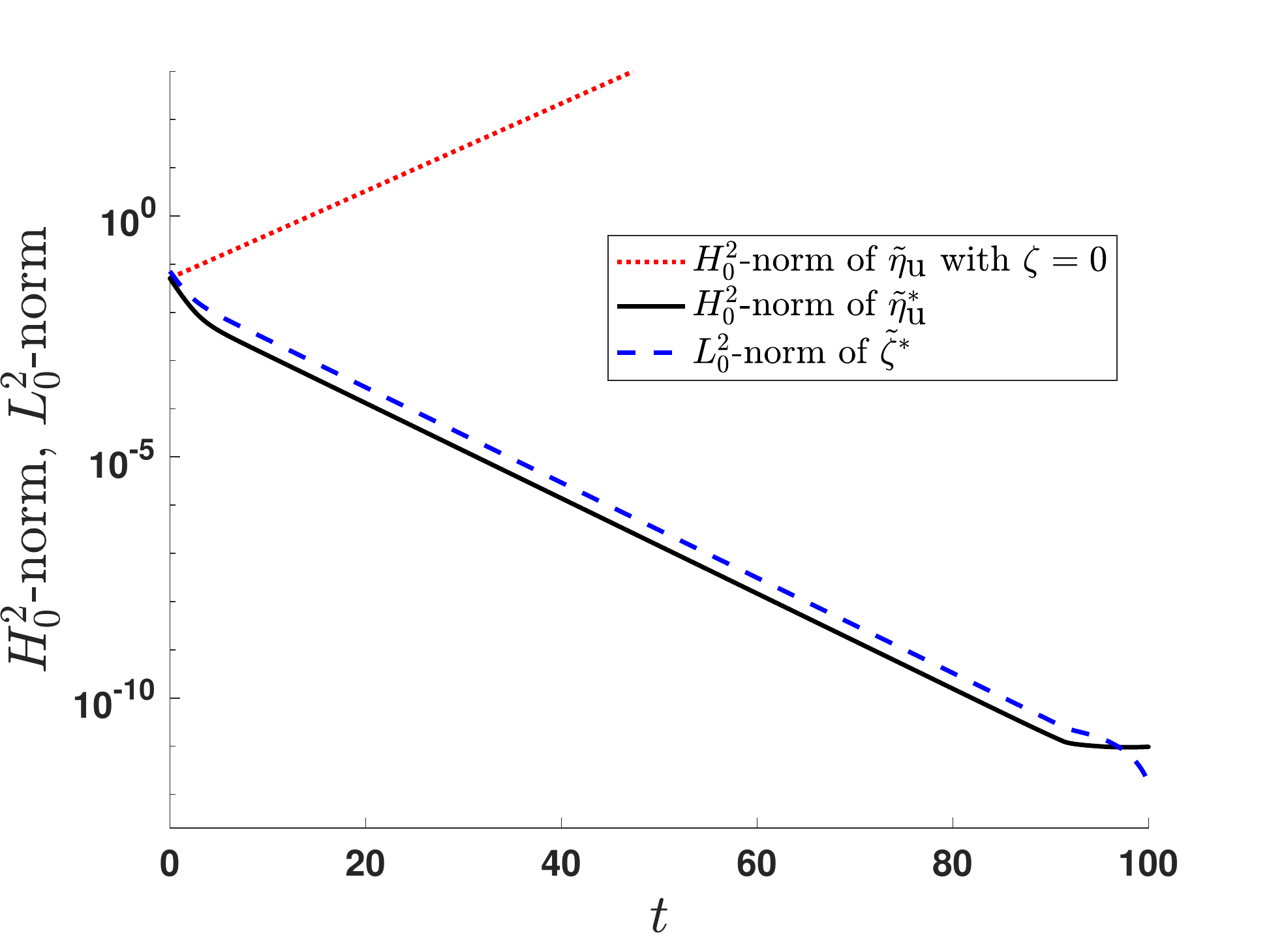}
\end{subfigure}
\begin{subfigure}{2.9in}
\caption{Components of $\mathcal{C}_{2,1}^*$ for case (ii).} 
\includegraphics[width=2.9in]{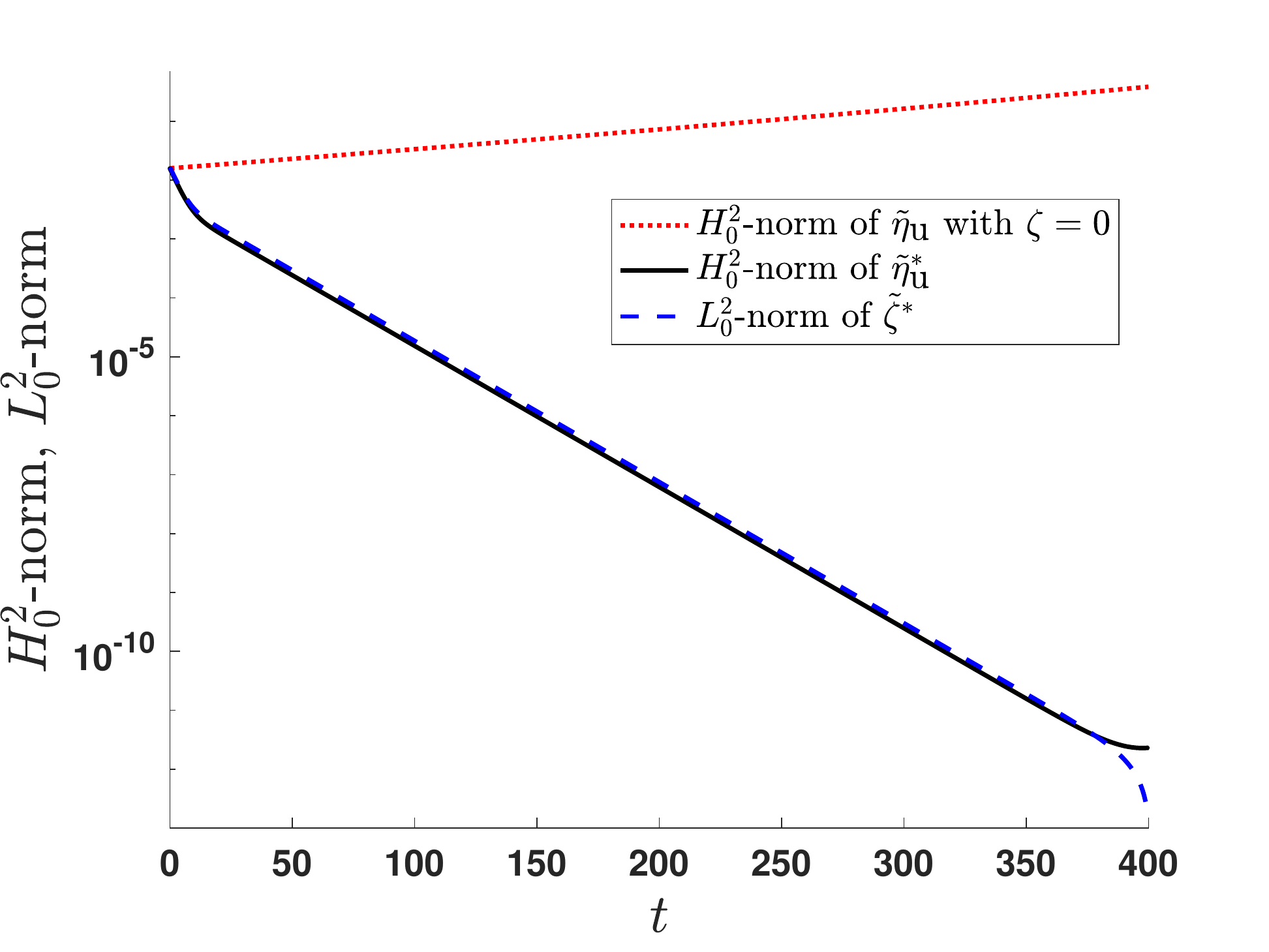}
\end{subfigure}
\begin{subfigure}{2.9in}
\caption{Profile of $\eta^*$ at $T=100$ for case (i).} 
\includegraphics[width=2.9in]{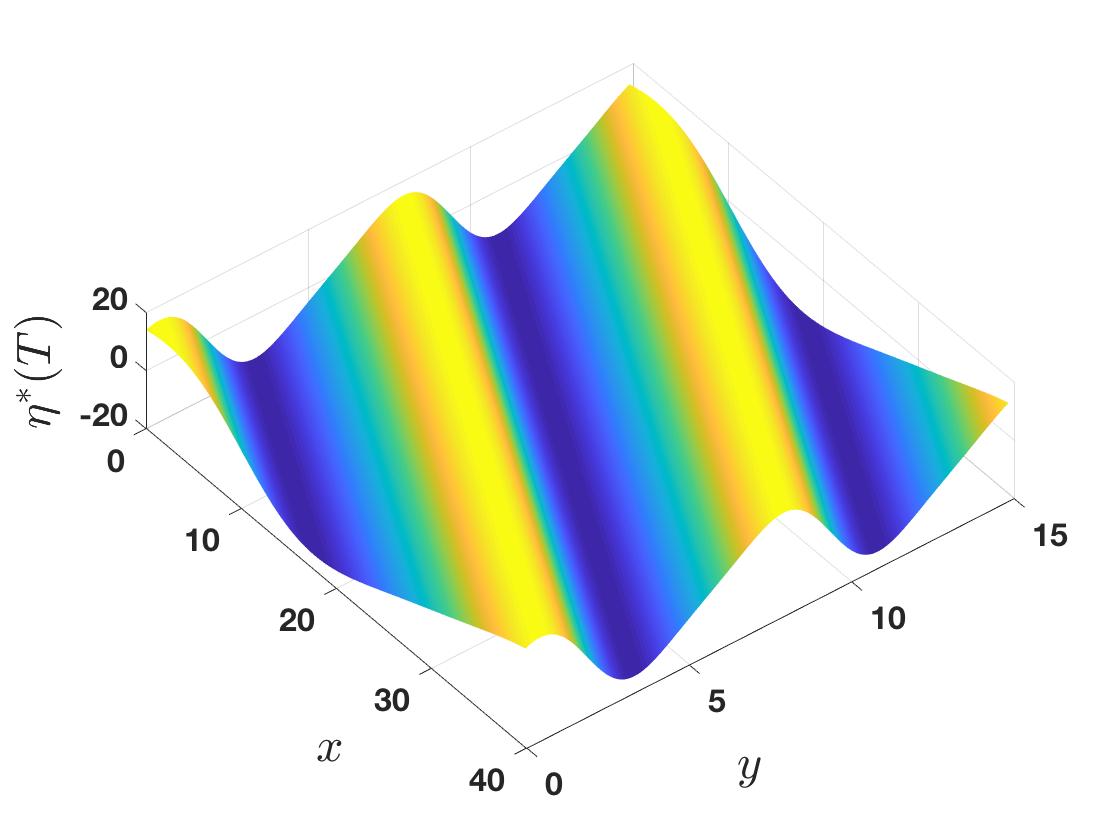}
\end{subfigure}
\begin{subfigure}{2.9in}
\caption{Profile of $\eta^*$ at $T=400$ for case (ii).} 
\includegraphics[width=2.9in]{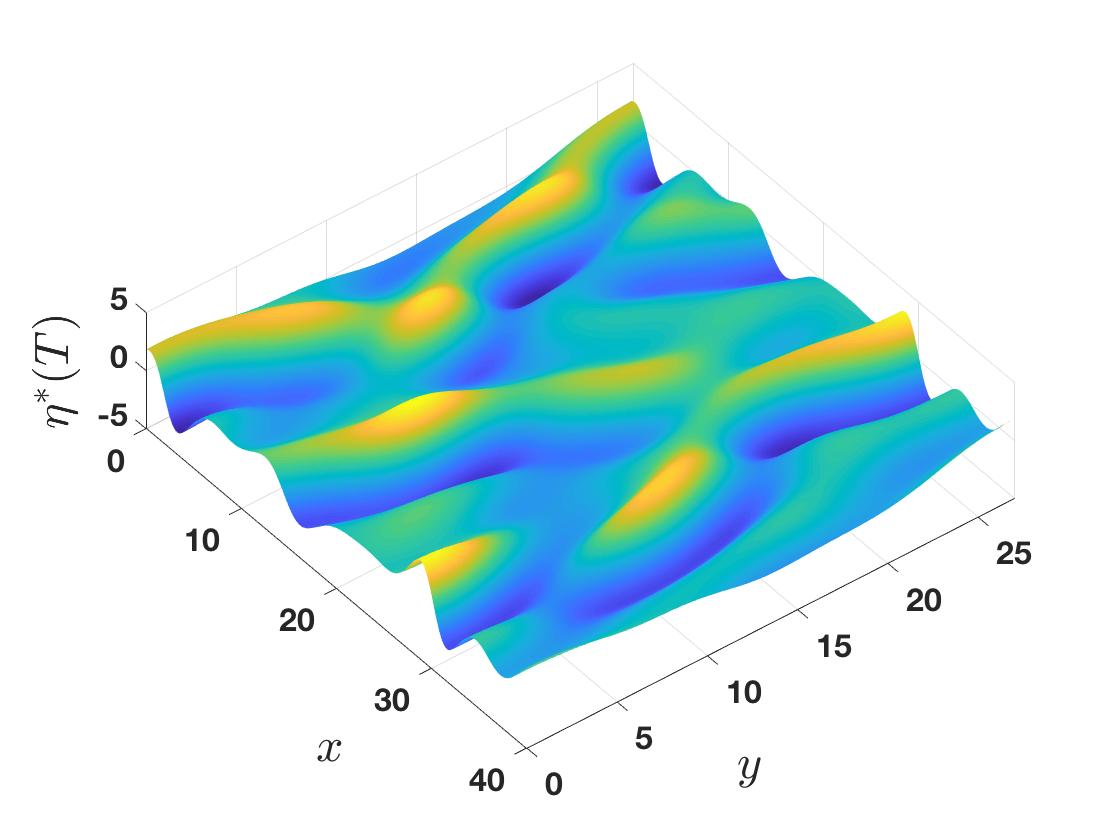}
\end{subfigure}
\caption{Application of optimal transverse controls for cases (i) (left) and (ii) (right) with $s=2$ and $\gamma = 1$. Panels (a,b) show the $L_0^2$-norms of the controlled and uncontrolled solutions (solid lines), as well as norms of the projections onto the subspaces spanned by the unstable transverse modes (dashed lines). Panels (c,d) show the behaviour of the integrands of the cost functional \eqref{transversecostfunc111}. The solution profiles at the respective final times are shown in panels (e,f).} \label{optimaltransversefig}
\end{figure}
Now we test the optimal transverse control constructed above numerically. So that unstable transverse modes are present, 
we consider two cases of hanging films, taking parameters (i) $\kappa = - 1 $, $L_1 = 40$, $L_2 = 15$ and $T = 100$, and (ii) $\kappa = - 0.25 $, $L_1 = 40$, $L_2 = 27$ and $T=400$. The final time $T$ differs between the two simulations as the timescale of the dynamics are dependent on $\kappa$, which varies between the two cases. For (i) and (ii), the $(0,\pm 1)$ and $(0,\pm 2)$-modes are unstable, and we take the following initial condition which contains contributions from these modes:
\begin{align}\label{initicond1}v(\bm{x}) = \frac{1}{10} & \left[ \cos\left(\frac{2\pi x}{L_1} \right) + \cos\left(\frac{2\pi x}{L_1} + \frac{2\pi y}{L_2}  \right) \right. \nonumber \\
&\qquad\qquad\quad\;\; \left. + \sin\left(\frac{4\pi x}{L_1} + \frac{2\pi y}{L_2}  \right) +  \sin\left( \frac{2\pi y}{L_2}  \right) + \sin\left( \frac{4\pi y}{L_2}  \right) \right].\end{align}
We choose $L_1$ to be sufficiently large that the dynamics of the controlled solutions are 
non-trivial, i.e. unstable streamwise and mixed modes are present. Figure \ref{optimaltransversefig} shows the optimal control successfully inhibiting the growth of transverse 
modes for cases (i) and (ii) in numerical simulations for the choices of $s = 2$ and $\gamma = 1$. For case (i), the uncontrolled solutions behaves as $\eta \sim \sin(4\pi y/L_2)e^{\lambda t},$
(the $(0,\pm2)$-modes are more unstable than the $(0,\pm1)$-modes) with exponential growth rate 
$\lambda = -\kappa(4\pi/L_2)^2 - (4\pi/L_2)^4 \approx 0.209$. With the application of the optimal transverse controls, the $H_0^2$-norm of the unstable modes and the $L_0^2$-norm of the control decay exponentially as shown in Figure \ref{optimaltransversefig}(c), and the energy of the full solution remains bounded as desired with a modal steady state (dominated by the $(m,2m)$-modes for $m\in\mathbb{Z}$) emerging. In case (ii), the $(0,\pm 1)$-modes dominate the uncontrolled solution with linear growth rate $\lambda \approx 0.0106$, and the controls successfully prevent unbounded growth and reveal a chaotic attractor for the uncontrolled dynamics. The cost functional takes the values $3.27\times 10^{-3}$ and $7.24\times 10^{-4}$ with the optimal control for the respective cases, in agreement with the analytical value given by \eqref{costreduced1}. We observe in both cases that the $H_0^2$-norm of the transverse modes and the $L_0^2$-norm of the optimal controls decay exponentially for the majority of the time interval $[0,T]$. For mid-range times, the controlled transverse modes decay with exponential decay rate proportional to the uncontrolled linear growth rate, and $r_{k_2} \approx \lambda_{-}$ (as in the infinite time horizon problem). For example, in case (i), the most unstable $(0,\pm2)$-modes are controlled to zero by the optimal control with a greatest decay rate. 
\begin{figure}
\centering
\includegraphics[width=2.9in]{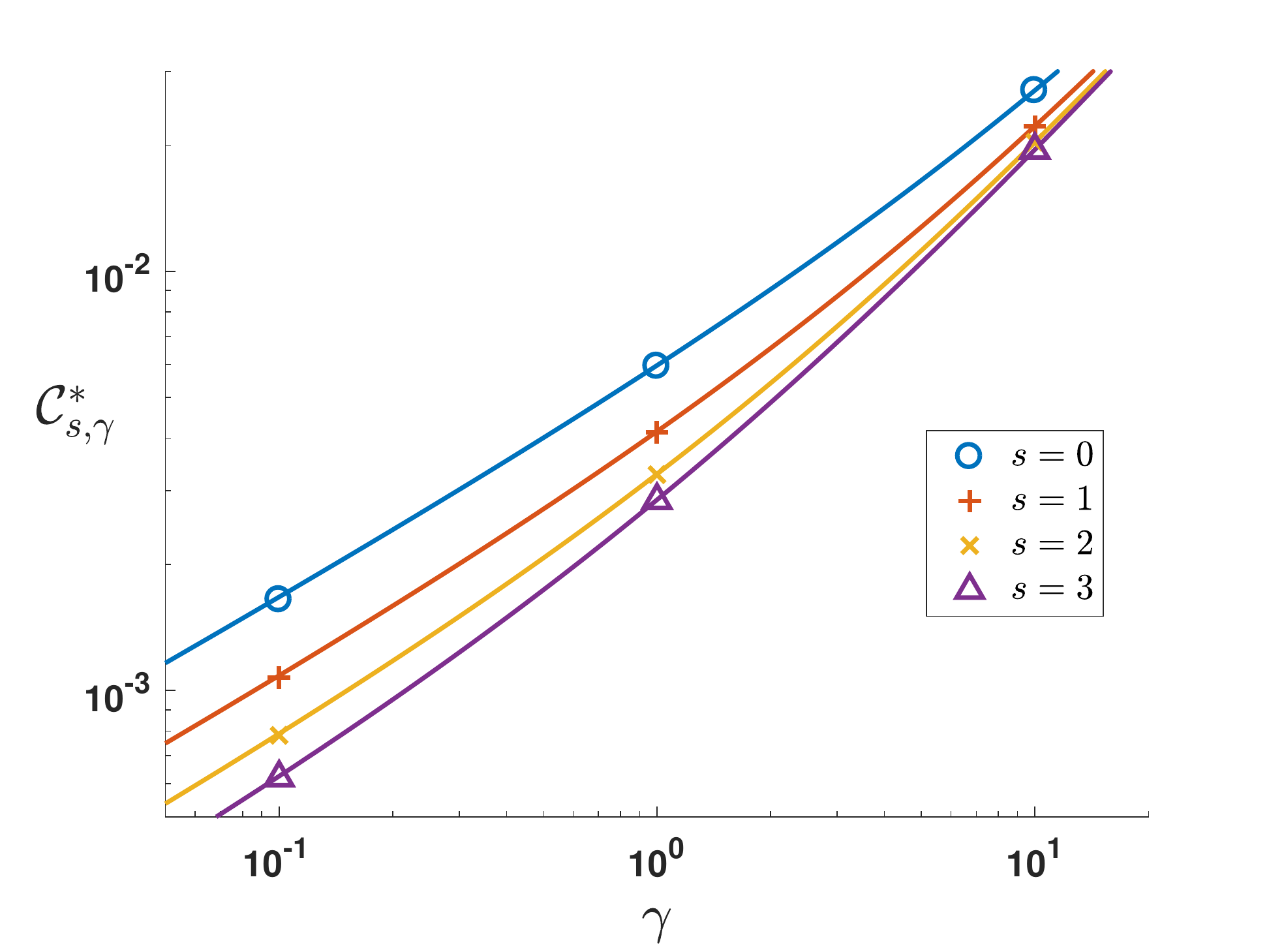}
\caption{Values of $\mathcal{C}_{s,\gamma}^*$ for case (i) as $\gamma$ and $s$ vary. The value of the cost functional attained by the optimal control and the corresponding solution is displayed for a range of $\gamma$ and $s\in \{0,1,2,3\}$, this is given analytically by \eqref{costreduced1}. The data points are obtained from numerical simulations for $\gamma\in \{0.1,1,10\}$, and give good agreement with the analytical values represented by the lines.}\label{costfunscatter}
\end{figure}
Figure \ref{costfunscatter} shows how the (analytical) optimal value of the the cost functional \eqref{costreduced1} is affected by changes in $\gamma$ and $s$ for case (i); we also validate our numerical results by superimposing the cost functional values obtained for simulations taking $\gamma \in \{ 0.1,1,10\}$ and $s\in \{ 0,1,2,3\}$. We observe that $\mathcal{C}_{s,\gamma}^*$ is an increasing function of $\gamma$ and a decreasing function of $s$ -- this is not surprising given the definitions of the $H_0^s$-norm \eqref{Hsnorm1aaa} and the cost functional \eqref{transversecostfunc111}. These cases all provide plots similar to those found in Figure \ref{optimaltransversefig} for case (i).

It now remains to comment briefly on the eventual dynamics of the system with $\tilde{\eta} = 0$ resulting from application of the above controls over a large time interval. Interestingly, we found that as the gravitational instability is strengthened ($\kappa$ decreased), the usual chaotic dynamics observed for the vertical film case in \cite{tkp2018} give way to diagonal modal attractors for the controlled solutions; this can partly be seen in the cases plotted in Figure \ref{optimaltransversefig}. It would be expected that a forwards Feigenbaum cascade occurs for overlying films as $\kappa$ decreases to zero (for fixed $L_1$ and $L_2$) given the results of Smyrlis and Papageorgiou \cite{depapa2} for the 1D problem, thus films close to vertical (controlled if necessary) are the most prone to chaotic dynamics. With the length parameters $L_1 = 40$ and $L_2 = 27$, as in case (ii), we find mostly chaotic dynamics as $\kappa$ is decreased to $-0.44$, below which a time-periodic attractor emerges. This behaviour dominates until $\kappa = -0.58$ when steady modal waves dominated by the $(m,3m)$-modes for $m\in\mathbb{Z}$ appear -- this is similar to the wave in Figure \ref{optimaltransversefig}(e).

As shown in this section, only the unstable transverse modes are responsible for the unbounded behaviour of \eqref{controlled2dks} with $\kappa < 0$. Nevertheless, this linear control theory may be extended to all transverse modes, where the summations in the previous lines extend from $\Xi$ to $\mathbb{Z}$ in order to drive the stable transverse modes to zero optimally, and $\tilde{\eta}_{\textrm{u}}$ or $\tilde{\eta}$ does not necessarily need to be controlled to zero (see appendix \ref{appendixdericoptlinear}). It is also worthwhile to note that the weakly nonlinear theory predicts that constant transverse forcing (as would arise from spanwise substrate corrugation) would not be a successful control method.

\section{Full optimal control\label{FOCref}} 

In this section, we consider the general optimal control problem for the 2D KSE \eqref{controlled2dks}, where the set of admissible controls $F_{\textrm{ad}}$ is a non-empty, closed, convex subset of $L^2(0,T;L_0^2)$. We consider the existence of an optimiser for our problem, and then give the methodology for numerical simulations. In contrast to section \ref{opttranssubsec1}, we must resort to an iterative algorithm for the full nonlinear problem. In the numerical simulations that follow the analysis, we take $F_{\textrm{ad}} = L^2(0,T;L_0^2)$.

Following the abstract formulation in \cite{lions1969quelques,temam1997infinite}, for example, we have the following local existence and uniqueness theorem (as in Pinto \cite{pinto1} for the case of $\zeta =0$ and $\kappa = 0$):
\begin{theorem}\label{existenceunThm1}For initial condition $v \in L_0^2$ and $\zeta \in L^2(0,T;L_0^2)$, there exists a unique solution $\eta$ to \eqref{controlled2dks} in $X_1:= C^0([0,T]; L_0^2) \bigcap L^2(0,T; H_0^2)$ (weak solution). Moreover, if $v\in H_0^2$, then $\eta \in X_2:= C^0([0,T]; H_0^2) \bigcap L^2(0,T; H_0^4)$ (strong solution).
\end{theorem}
We shall use both parts of the above theorem for the proof of existence of an optimal control. Uniqueness of an optimal control is not guaranteed as the optimisation problem is not convex -- this is due to the $\eta\eta_x$ nonlinearity. The following theorem and proof makes no deep assumptions (such as analyticity or regularity in higher Sobolev spaces) which would be expected to be true for equations such as \eqref{controlled2dks}, only requiring the above existence and uniqueness theorem. Thus, the proof is given in a very general framework which may be applied to similar problems. However, this restricts us in the range of the index $s$, and the regularity of the initial condition and desired state.

\begin{theorem}\label{thm4.2} Let $v \in H_0^2$, $s \leq 2$, $\gamma > 0$, and take $\overline{\eta} \in C^0([0,T];H_0^s)$. Define $F_{\textrm{ad}}$ to be a non-empty, closed, convex subset of $L^2(0,T;L_0^2)$. Then, there exists an optimal control $\zeta^*\in F_{\textrm{ad}}$ for the KSE \eqref{controlled2dks} with initial condition $v$ which minimises the cost functional $\mathcal{C}_{s,\gamma}$ defined by \eqref{costfunctional1aa}.
\end{theorem}
\begin{proof}
For a control $\zeta$, we write the solution of \eqref{controlled2dks} with the given initial condition $v$ in terms of the control-to-state map as $\eta(\zeta;v)$, through which we define the reduced cost functional $\tilde{\mathcal{C}}_{s,\gamma}(\zeta) = \mathcal{C}_{s,\gamma}(\eta(\zeta;v),\zeta;\overline{\eta})$. Since $v \in H_0^2 \subset L_0^2$, the first part of Theorem \ref{existenceunThm1} implies that $(\eta(\zeta;v),\zeta) \in X_1 \times F_{\textrm{ad}}$. The optimal control problem can thus be recast as the problem of finding the minimiser of $\tilde{\mathcal{C}}_{s,\gamma}(\zeta)$ over $F_{\textrm{ad}}$. It makes sense to check that there exist $\zeta \in F_{\textrm{ad}}$ which give a finite value of $\tilde{\mathcal{C}}_{s,\gamma}$. The boundedness of the first and last terms of the cost functional \eqref{costfunctional1aa} are consequences of $(\eta(\zeta;v),\zeta) \in X_1 \times F_{\textrm{ad}}$ (recall that $s\leq 2$) and the regularity of $\overline{\eta}$. The final time component is the problematic term which requires the additional regularity of the initial condition: by the second part of Theorem \ref{existenceunThm1}, we know that $\eta|_{t=T} \in H_0^2 \subseteq H_0^s$, and since $\overline{\eta}|_{t=T} \in H_0^s$, the final time component of \eqref{costfunctional1aa} is also finite. Thus any forcing in the space of admissible controls yields a finite cost. Since $\mathcal{C}_{s,\gamma} \geq 0$, the reduced cost has a finite infimum,
\begin{equation}\inf_{\zeta \in F_{\textrm{ad}}} \tilde{\mathcal{C}}_{s,\gamma}(\zeta) = c \geq 0.\end{equation}
We cannot yet say that this infimum is attained by a control in the admissible space, however, we know that there is a minimising sequence $\{\zeta^{(n)}\}_{n=1}^{\infty} \subseteq F_{\textrm{ad}}$ with associated states defined by $\eta^{(n)} = \eta( \zeta^{(n)};v)$ such that
\begin{equation}\lim_{n \rightarrow \infty} \tilde{\mathcal{C}}_{s,\gamma}(\zeta^{(n)}) =  c.\end{equation}
Without loss of generality, we may assume that the corresponding sequence of costs is monotonically decreasing. From the form of the cost functional \eqref{costfunctional1aa}, we know that $\{\zeta^{(n)}\}_{n=1}^{\infty}$ is bounded uniformly in $L^2(0,T;L_0^2)$, i.e. there is some constant $r \geq 0$ ($ r = \tilde{\mathcal{C}}_{s,\gamma}(\zeta^{(1)}) $ suffices) such that
\begin{equation} \label{zetaball1} \| \zeta^{(n)} \|_{L^2(0,T;L_0^2)} \leq r.
\end{equation}
We define $K$ to be the intersection of $F_{\textrm{ad}}$ with the closed ball of radius $r$ in $L^2(0,T;L_0^2)$. $K$ is a closed, convex and bounded subset of the reflexive Banach space $L^2(0,T;L_0^2)$, and thus is weakly sequentially compact (see Thm 2.10 and Thm 2.11 in \cite{Trol1}). Then, the sequence $\{\zeta^{(n)}\}_{n=1}^{\infty} \subseteq K$ has a weakly convergent subsequence $\zeta^{(n)} \rightharpoonup \zeta^*$ for the topology of $L^2(0,T;L_0^2)$ (not relabelled for simplicity) with weak limit $\zeta^* \in K \subseteq F_{\textrm{ad}}$. For more general forms of cost functional, it may be necessary to assume that $F_{\textrm{ad}}$ is bounded to make this step. The function $\zeta^*$ is a candidate for the optimal control.

Multiplying the KSE \eqref{controlled2dks} by the solution and taking the spatial average, we obtain the inequality (see appendix \ref{Appendixboundderiv} for the derivation)
\begin{equation}\label{importantbound} \| \eta^{(n)} \|_{C^0([0,T]; L_0^2)} + \| \eta^{(n)} \|_{L^2(0,T;H_0^2)}  \leq C (  \| v \|_{L_0^2}  +  \|  \zeta^{(n)} \|_{L^2(0,T; L_0^2)}),
\end{equation}
where the constant $C$ depends on $\kappa$ and $T$ only. With this estimate and \eqref{zetaball1}, the sequence $\{\eta^{(n)}\}_{n=1}^{\infty}$ is uniformly bounded in $X_1$. Since $X_1$ is the dual of a separable Banach space\footnote{The predual of $X_1$ is the direct sum space $L^1(0,T; (L_0^2)^*) + L^2(0,T;(H_0^2)^*)$ where stars denote duals -- this is the direct sum of two separable Banach spaces and is thus separable Banach space itself (endowed with an appropriate norm).},
with the Banach--Alaoglu theorem, there is a subsequence (not relabelled for simplicity) with
\begin{equation}\label{weakconvresult1}\eta^{(n)} \rightharpoonup \eta^* \; \textrm{weakly in} \; L^2(0,T;H_0^2), \qquad \eta^{(n)} \rightharpoonup \eta^* \; \textrm{weak-star in} \; L^{\infty}(0,T;L_0^2).\end{equation}
It also follows from (\ref{zetaball1},\ref{importantbound}) that the sequence with terms
\begin{equation}\label{linearseq1}( \kappa - 1) \eta_{xx}^{(n)}  + \kappa  \eta_{yy}^{(n)} - \Delta^2 \eta^{(n)} \end{equation}
is uniformly bounded in $L^2(0,T,H_0^{-2})$ -- a Poincar\'{e} inequality for the continuous embedding of $H_0^2$ in $L_0^2$ is also utilised. To deal with the nonlinearity, we use the estimate
\begin{equation} \| \eta^{(n)} \eta^{(n)}_x \|_{H_0^{-1}} \leq \frac{1}{2} \| (\eta^{(n)})^2 \|_{L_0^{2}} \leq \frac{1}{2} \| \eta^{(n)} \|_{L^{\infty}} \| \eta^{(n)} \|_{L_0^{2}} \leq \widehat{C}\| \eta^{(n)} \|_{H_0^2} \| \eta^{(n)} \|_{L_0^{2}},\end{equation}
where the first inequality follows from $\eta^{(n)} \eta^{(n)}_x = \partial_x (\eta^{(n)})^2/2$ and the definition of the $H_0^{-1}$-norm (\ref{Hsnorm1aaa}a), and the last uses an Agmon inequality proved in Thm 4.1 of \cite{ilyin2006sharp} and again a Poincar\'{e} inequality for the embedding of $H_0^2$ in $L_0^2$. The constant $\widehat{C}$ involves the constants from both the Agmon and Poincar\'{e} inequalities. Squaring and integrating in time leads to
\begin{equation} \| \eta^{(n)} \eta^{(n)}_x \|_{L^2(0,T;H_0^{-1})}^2  \leq \widehat{C}^2 \int_0^T \| \eta^{(n)} \|_{H_0^2}^2 \| \eta^{(n)} \|_{L_0^{2}}^2 \; \mathrm{d}t \leq \widehat{C}^2C^4 (  \| v \|_{L_0^2}  +  r)^4,\end{equation}
where we have used (\ref{zetaball1},\ref{importantbound}). Then the sequence $\{ \eta^{(n)} \eta^{(n)}_x \}_{n=1}^{\infty}$ is bounded uniformly in $L^2(0,T,H_0^{-1})$, thus bounded uniformly in $L^2(0,T,H_0^{-2})$ again by a Poincar\'{e} inequality. With these results and the uniform boundedness of $\{\zeta^{(n)}\}_{n=1}^{\infty}$, we can conclude that $\{ \eta^{(n)}_t \}_{n=1}^{\infty}$ is bounded uniformly in $L^2(0,T,H_0^{-2})$. With this and the weak convergence result (\ref{weakconvresult1}a), it follows from Thm 8.1 in \cite{robinson2001infinite} (a compactness result) that
\begin{equation}\label{strongconvresult1}\eta^{(n)} \rightarrow \eta^* \; \textrm{strongly in} \; L^2(0,T;H_0^s), \textrm{ for } s < 2.\end{equation}

We now check that the weak limits satisfy the control-to-state map, $\eta^* = \eta(\zeta^*;v)$ -- the pair $(\eta^*,\zeta^*)$ must satisfy the governing equation \eqref{controlled2dks}, and also must satisfy the initial state, $\eta^*|_{t=0} = v$. The strong convergence result \eqref{strongconvresult1} with $s=1$ implies that the nonlinearity converges weakly in $L^2(0,T;(H_0^{2})^*)$, where the star denotes the dual, as follows: Let $w \in L^2(0,T;H_0^2)$, then
\begin{align}\nonumber \left| \int_0^T \int_Q (\eta^{(n)} \eta^{(n)}_x - \eta^* \eta^*_x ) w \;\mathrm{d}\bm{x} \; \mathrm{d}t \right| = & \;  \frac{1}{2} \left| \int_0^T \int_Q (\eta^{(n)} + \eta^* )  (\eta^{(n)} - \eta^* ) w_x \;\mathrm{d}\bm{x} \; \mathrm{d}t \right| \\
\nonumber \leq & \; \frac{1}{2}  \| \eta^{(n)} + \eta^* \|_{C^0([0,T];L_0^2)}  \| \eta^{(n)} - \eta^* \|_{L^2(0,T;L_0^4)}    \|  w_x \|_{L^2(0,T;L_0^4)} \\
\leq & \; \tilde{C}^2 C (  \| v \|_{L_0^2}  +  r)  \| \eta^{(n)} - \eta^* \|_{L^2(0,T;H_0^1)}    \|  w \|_{L^2(0,T;H_0^2)} \label{nonlinweakconverge222}
\end{align}
where $\tilde{C}$ is the constant corresponding to the continuous embedding of $H_0^1$ in $L_0^4$ in two space dimensions. The right hand side of \eqref{nonlinweakconverge222} converges to zero as $n \rightarrow \infty$, and thus 
\begin{equation} \eta^{(n)}  \eta^{(n)}_x \rightharpoonup \eta^* \eta^*_x \; \textrm{weakly in} \; L^2(0,T;(H_0^{2})^*).\end{equation}
Also, the sequence with terms \eqref{linearseq1} converges weakly to its corresponding optimal limit in the same space. Since $\{ \eta^{(n)}_t \}_{n=1}^{\infty}$ is bounded uniformly in $L^2(0,T,H_0^{-2}) \subset L^2(0,T,(H_0^{2})^*)$, it is weakly convergent, and furthermore by a density argument we can conclude that the weak limit is $\eta_t^*$. Since the sequence 
$\{\zeta^{(n)}\}_{n=1}^{\infty}$ converges weakly in $L^2(0,T;L_0^2) \subset L^2(0,T,(H_0^{2})^*)$, by uniqueness of weak limits we may conclude that
\begin{equation}
\eta_t^* + \eta^* \eta^*_x + (1- \kappa) \eta^*_{xx}  -  \kappa  \eta^*_{yy} + \Delta^2 \eta^*  =  \zeta^*,
\end{equation}
holds in the $L^2(0,T,(H_0^{2})^*)$-sense. The proof that the optimal state has initial condition $v$ follows similarly to the arguments in Thm 9.3 of \cite{robinson2001infinite}.

Lastly we show that the pair $(\eta^*,\zeta^*)$ is a minimiser of the cost functional. Since $\zeta^* \in F_{\textrm{ad}}$ and $v \in H_0^2$, we have the higher regularity $\eta^* \in X_2$. The weak lower semicontinuity of the individual components of the cost \eqref{costfunctional1aa} in the state and control respectively then yields
\begin{equation}\inf_{\zeta \in F_{\textrm{ad}}} \tilde{\mathcal{C}}_{s,\gamma}(\zeta) =  c = \lim_{n \rightarrow \infty} \mathcal{C}_{s,\gamma}(\eta^{(n)},\zeta^{(n)};\overline{\eta}) = \mathcal{C}_{s,\gamma}(\eta^*,\zeta^*;\overline{\eta}).\end{equation}
\end{proof}
We remark that the above theorem holds for $v \in L_0^2$ and $s \leq 0$ by a similar proof, and extension to any $s\in\mathbb{R}$ would be viable with analyticity results.

We now present the adjoint framework which allows the construction of such minimising sequences; this forms the basis of the iterative algorithm we will employ for our numerical simulations. The Lagrangian of the optimisation problem is
\begin{equation}\label{Lagrangian1aa}\mathcal{L}(\eta,\overline{\eta},\zeta, p) = \mathcal{C}_{s,\gamma}(\eta,\zeta;\overline{\eta}) - \left\langle p, \eta_t + \eta \eta_x + (1- \kappa) \eta_{xx}  -  \kappa  \eta_{yy} + \Delta^2 \eta  -  \zeta \right\rangle_{L^2(0,T;L_0^2)},
\end{equation}
where the $L_0^2$-inner product corresponds to the energy density norm defined by (\ref{Hsnorm1aaa}b). Here, $p$ is known as the adjoint variable which in effect is a Lagrange multiplier. The first-order conditions of optimality for a local optimiser $(\eta^*, \zeta^*, p^*)$ consist of the governing equation \eqref{controlled2dks}, an adjoint equation, and a variational inequality. The governing equation arises from the functional (Fr\'{e}chet) derivative of $\mathcal{L}$ with respect to the adjoint variable $p$, and taking the functional derivative of the Lagrangian with respect to $\eta$ gives the adjoint equation
\begin{equation}\label{Adjointeqn1} - p_t - (I - P_{\bm{0}})\eta p_x + (1-\kappa) p_{xx} - \kappa p_{yy} + \Delta^2 p = (-\Delta)^s ( \eta - \overline{\eta} )\end{equation}
where $I - P_{\bm{0}}$ is the projection onto the space of functions with zero mean and $(-\Delta)^s$ is the fractional Laplacian of order $2s$ with Fourier symbol $| \bm{\tilde{k}} |^{2s}$. Both \eqref{controlled2dks} and \eqref{Adjointeqn1} are necessarily satisfied by a local optimiser $(\eta^*, \zeta^*, p^*)$. The adjoint equation is backwards in time, and is supplied with the final time condition 
\begin{equation}\label{AdjointeqnfinaltimeBC}p(\bm{x},T) = (-\Delta)^s(\eta(\bm{x},T) - \overline{\eta}(\bm{x},T)).\end{equation}
Lastly, we obtain a variational inequality by taking the Fr\'{e}chet derivative with respect to $\zeta$,
\begin{equation}\label{variationineq1} \left\langle \gamma \zeta^* + p^*, \zeta - \zeta^*\right\rangle_{L^2(0,T;L_0^2)} \geq 0, \quad \forall \zeta \in F_{\textrm{ad}}.\end{equation}
This inequality informs us that the optimal control $\zeta^*$ is in the direction of $-(\gamma\zeta + p)$ from the current control $\zeta$, from which we may construct iteration schemes. Such updates move along the local approximation to the curve of steepest descent. The exclusion of the zero mode, $\bm{k} = 0$, in the $L^2(0,T;L_0^2)$-inner product is vital to ensuring that the control remains in the space of zero average functions. The projection $I - P_{\bm{0}}$ appearing in \eqref{Adjointeqn1} guarantees that the spatial average of the adjoint variable is fixed at zero ($\eta p_x$ is not a zero mean function in general). Lastly, the update direction $- (\gamma\zeta + p)$ also has zero spatial average, and thus any iteration scheme will yield a sequence of controls preserving this property.

Next we detail the numerical algorithm we employ to approximate a local optimiser $(\eta^*, \zeta^*, p^*)$; checking that this is the global optimiser for our infinite-dimensional problem is difficult. In simple terms, the forward--backward sweep method comprises of iterated simulations of \eqref{controlled2dks} and \eqref{Adjointeqn1}, with control updates after each iteration. Note that \eqref{Adjointeqn1} may be written in the form
\begin{equation}- p_t + \mathcal{A}p = \mathcal{B}'(p;\eta,\overline{\eta}), \qquad \mathcal{B}'(p;\eta,\overline{\eta}) = (I - P_{\bm{0}}) \eta p_x + (-\Delta)^s(\eta - \overline{\eta}) + cp,\end{equation}
where $\mathcal{A}$ is defined by (\ref{AandBdefns1}a). Thus, the BDF methods outlined in subsection~\ref{sec:KSE} are applicable to this backwards in time equation with the same value of $c$, provided that $\mathcal{B}'$ satisfies the necessary Lipschitz bounds. The pseudocode for the forward--backward sweeping method is given in Algorithm \ref{alg:forback}.
\begin{algorithm}[h]
\caption{Forward--backward sweep method with adaptive step-halving.}
\label{alg:forback}
\begin{algorithmic}
\STATE{Choose initial state $\eta_0(\bm{x})$, desired target state $\overline{\eta}(\bm{x},t)$ and initial control guess $\zeta^{(1)}(\bm{x},t)$ for $t\in [0,T]$. Take $c^{(1)} \in (0,1)$, set $c^{(n)} = c^{(1)}$ for $n \in \mathbb{N}_{>1}$, and choose a tolerance $\tau$.}
\STATE{Initialise with $n = 0$, $\mathcal{C}_{s,\gamma}(0) = 0$, $d = \tau + 1$.}
\WHILE{$d > \tau$}
\STATE{$n = n+1$.}
\STATE{Solve \eqref{controlled2dks} with initial state $\eta_0$ and control $\zeta^{(n)}$ to obtain $\eta^{(n)}(\bm{x},t)$ for $t\in [0,T]$.}
\STATE{Calculate $\mathcal{C}_{s,\gamma}(n):=\mathcal{C}_{s,\gamma}(\eta^{(n)},\zeta^{(n)};\overline{\eta})$ and update $d = |\mathcal{C}_{s,\gamma}(n)-\mathcal{C}_{s,\gamma}(n-1)|$.}
\STATE{Solve \eqref{Adjointeqn1} backwards given final time condition \eqref{AdjointeqnfinaltimeBC} and state $\eta^{(n)}$ to obtain $p^{(n)}(\bm{x},t)$ for $t\in [0,T]$. }
\STATE{Update control $\zeta^{(n+1)} =(1 - c^{(n)})\zeta^{(n)} - c^{(n)} p^{(n)}/\gamma$. }
\STATE{Set $\mathcal{C}_{s,\gamma}(n+1) = \mathcal{C}_{s,\gamma}(n)+1$.}
\WHILE{$\mathcal{C}_{s,\gamma}(n+1) \geq \mathcal{C}_{s,\gamma}(n)$}
\STATE{Solve \eqref{controlled2dks} with initial state $\eta_0$ and control $\zeta^{(n+1)}$ to obtain $\eta^{(n+1)}(\bm{x},t)$ for $t\in [0,T]$. Calculate $\mathcal{C}_{s,\gamma}(n+1):=\mathcal{C}_{s,\gamma}(\eta^{(n+1)},\zeta^{(n+1)};\overline{\eta})$.}
\IF{$\mathcal{C}_{s,\gamma}(n+1) \geq \mathcal{C}_{s,\gamma}(n)$}
\STATE{$c^{(m)} = c^{(m)}/2$ for $m \geq n$.}
\STATE{Update control $\zeta^{(n+1)} =(1 - c^{(n)})\zeta^{(n)} - c^{(n)} p^{(n)}/\gamma$.}
\ENDIF
\ENDWHILE
\ENDWHILE
\RETURN $\eta^{(n)}$, $\zeta^{(n)}$.
\end{algorithmic}
\end{algorithm}
The returned values $(\eta^{(n)},\zeta^{(n)})$ are a minimising sequence as in Theorem \ref{thm4.2}, with limit approximating the optimal state and control pair, $(\eta^*,\zeta^*)$. For our control update, we take a step of size $c^{(n)}/\gamma$ in the direction of $- (\gamma \zeta^{(n)} + p^{(n)})$,
\begin{equation}\zeta^{(n+1)} = \zeta^{(n)} - c^{(n)} ( \zeta^{(n)}\gamma  + p^{(n)}) /\gamma = (1 - c^{(n)})\zeta^{(n)} - c^{(n)} p^{(n)}/\gamma.\end{equation}
The above standard gradient descent can be rewritten as the convex combination of the current iterations of the control and adjoint variables. This classical method is found to perform consistently well compared to more complicated methods which may converge too quickly \cite{lenhart2007optimal}. Rather than performing an expensive line search across a range of convex combinations (i.e. seeking a minimum as $c^{(n)}$ is varied), we employ an adaptive step-halving scheme (a backtracking line search) for our simulations starting with steps of size $c^{(1)} = 0.1$ (see the update formula in Algorithm \ref{alg:forback}); if a control update results in an increased cost, then the step size (and the successive step sizes) are halved until the updated control yields a lower cost than the previous iteration. The sequence of values of the cost functional can be used to indicate a posteriori that we have achieved convergence to a local minimiser. We have checked the results of our numerical schemes with more complicated updating and searching methods with good agreement.


\begin{figure}[h!]
\centering
\begin{subfigure}{2.9in}
\caption{ Profile of $\overline{\eta}$.} 
\includegraphics[width=2.9in]{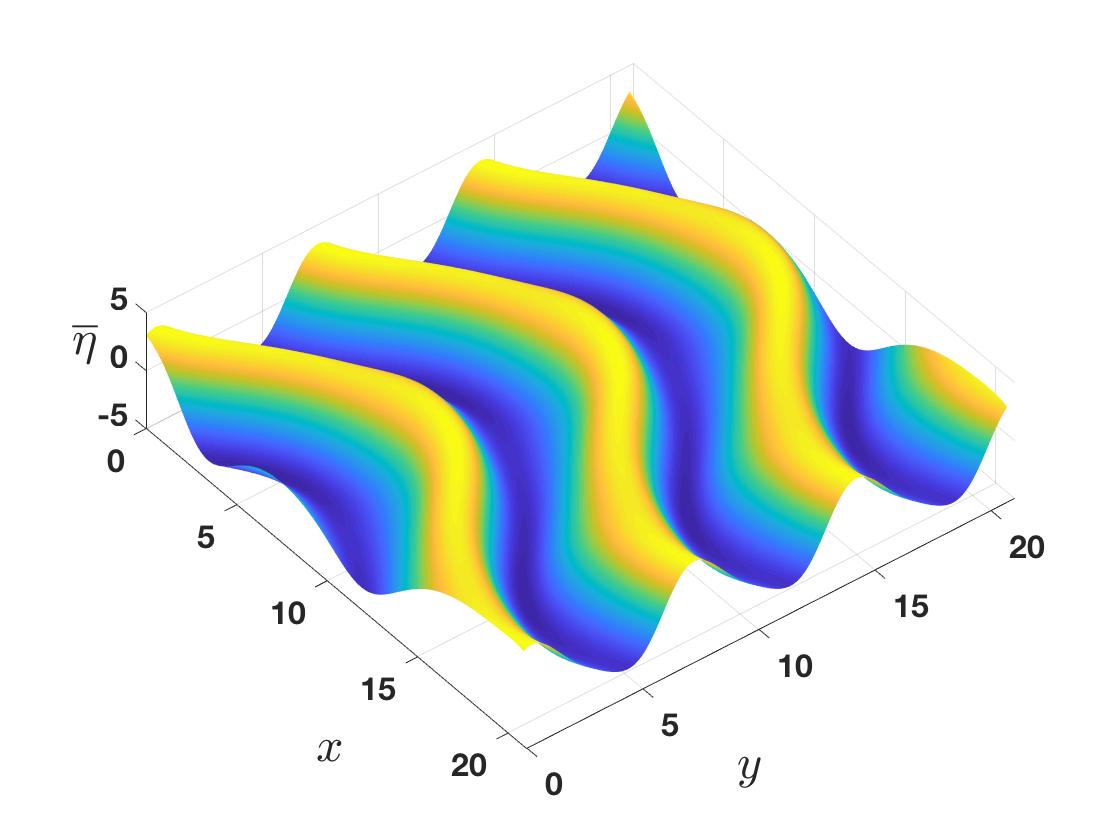}
\end{subfigure}
\begin{subfigure}{2.9in}
\caption{ Profile of $\eta^* - \overline{\eta}$ at $T = 5$.} 
\includegraphics[width=2.9in]{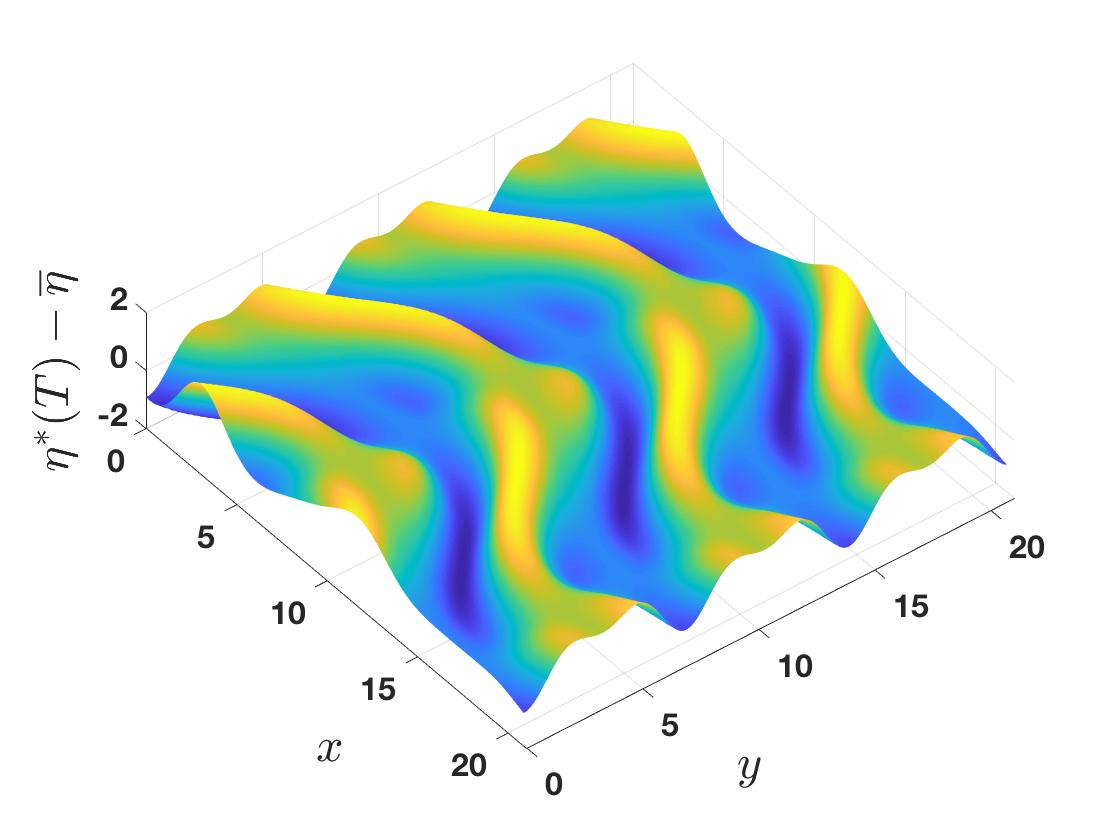}
\end{subfigure}
\begin{subfigure}{2.9in}
\caption{ Components of $\mathcal{C}_{0,1}^*$.} 
\includegraphics[width=2.9in]{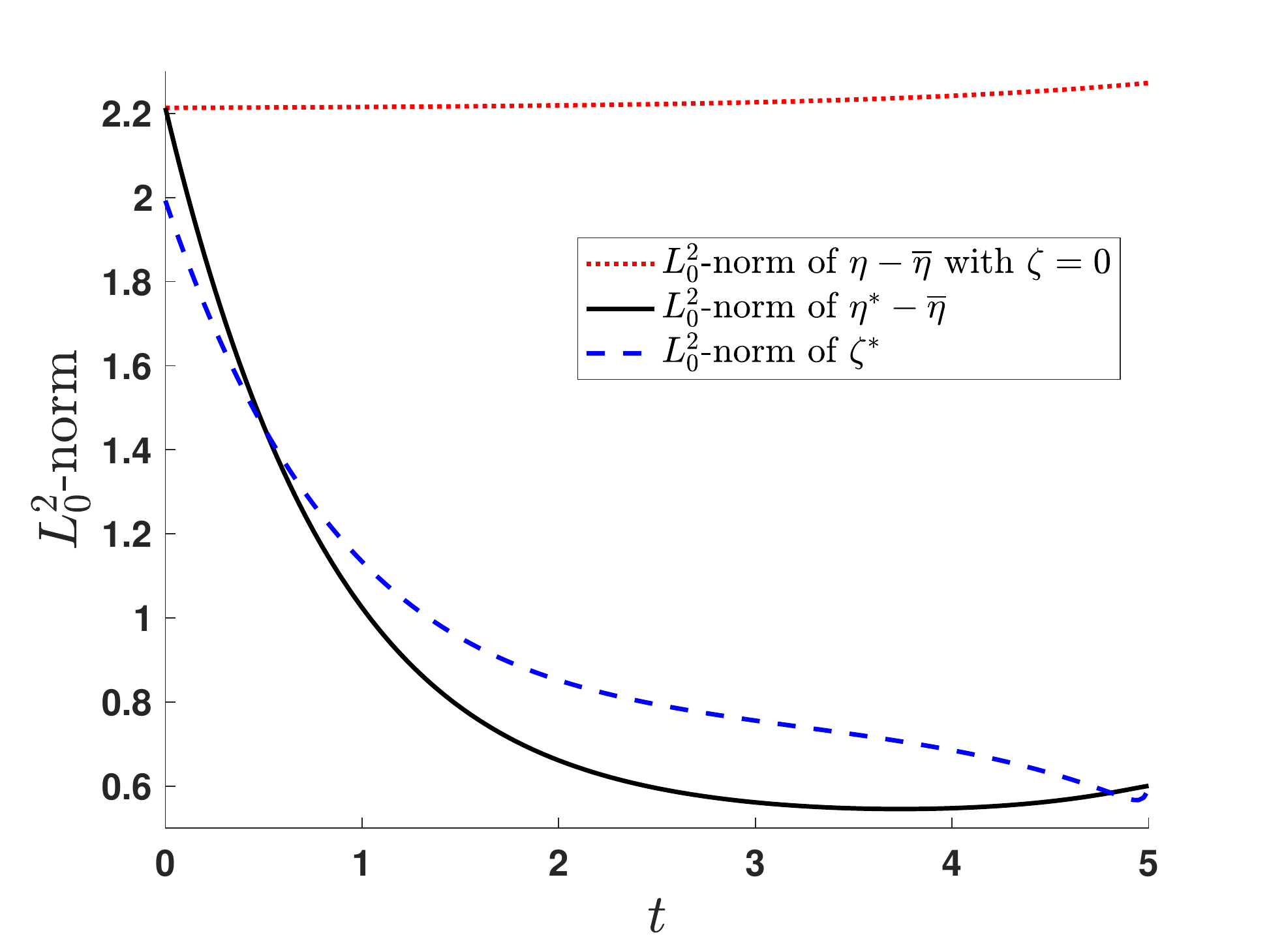}
\end{subfigure}
\begin{subfigure}{2.9in}
\caption{ Minimising sequence of costs.} 
\includegraphics[width=2.9in]{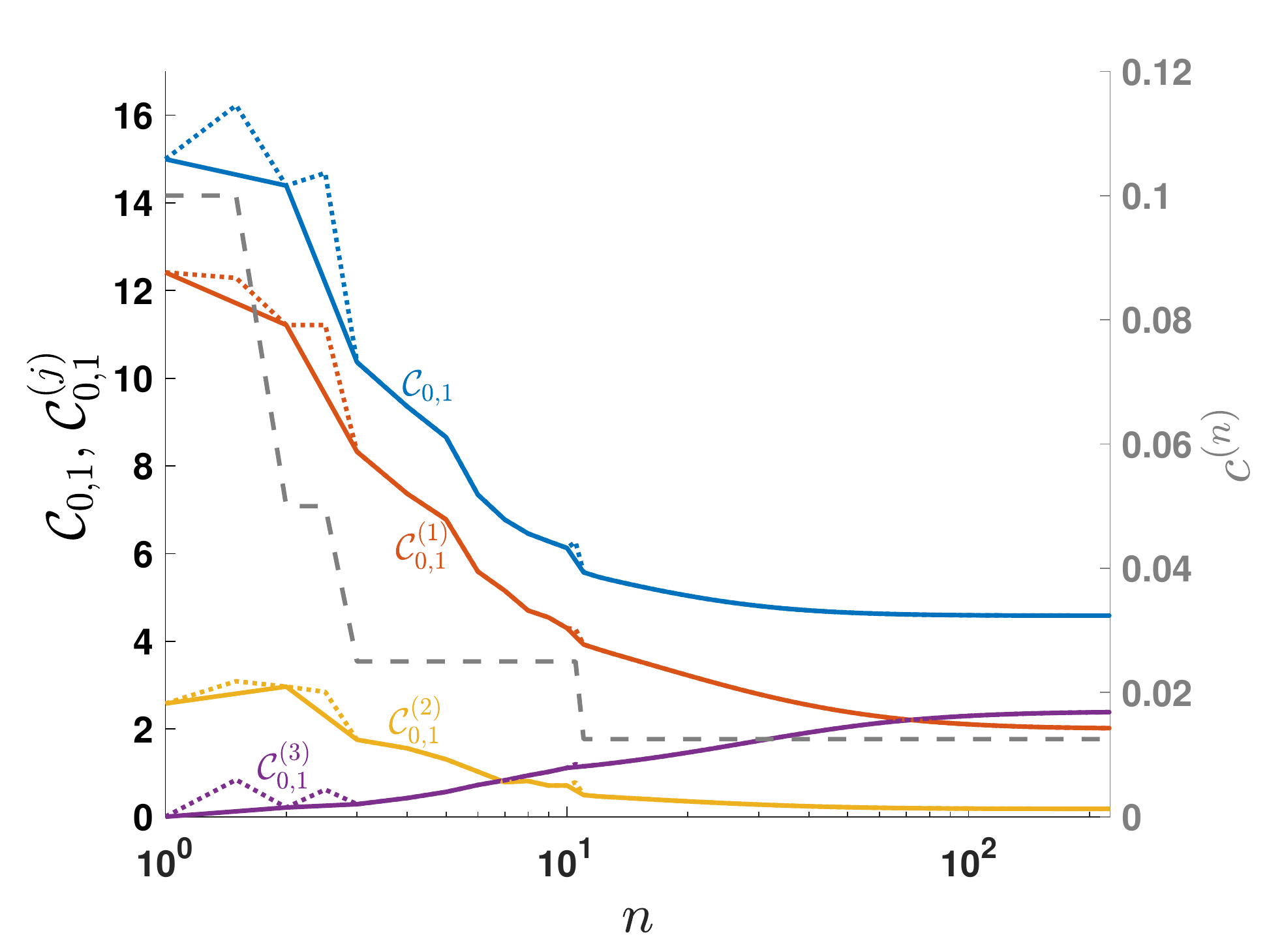}
\end{subfigure}
\caption{Full optimal control for a hanging film flow with $\kappa = -0.5$ and $L_1 = L_2 = 21$ with $s=0$ and $\gamma = 1$ over the time interval $[0,5]$ from initial condition \eqref{initicond1}. (a) The desired state, a ``snaking" transverse wave. (b) The difference between the optimal and desired states at the final time. (c) Components of the cost functional. The $L_0^2$-norms of $\eta - \overline{\eta}$ for both the uncontrolled (dotted line) and optimally controlled (solid line) cases. The $L_0^2$-norm of the optimal control is also included (dashed line). These are the integrands in $C_{0,1}^{(1)*}$ and $C_{0,1}^{(3)*}$. (d) The minimising sequence of costs (also broken down into its three components) against $n$ with solid lines. The dotted lines are with points at half-integers represent the gradient descent steps which resulted in larger costs, resulting in step-halving; the faded dashed line taking values on the right axis is the step $c^{(n)}$.} \label{optimalExp2fig}
\end{figure}
For the first numerical experiment, we take $L_1 = L_2 = 21$, and cover the three different dynamical regimes and two critical points of the system with $\kappa \in \{- 0.5,0, 0.5, 1,1.5\}$. We again take the initial condition defined by \eqref{initicond1}, and choose the desired state $\overline{\eta}$ to be the ``snaking" transverse wave shown in panel (a) of Figure \ref{optimalExp2fig}; this corresponds to a steady solution of the non-local 2D KSE studied in \cite{tomlin_papageorgiou_pavliotis_2017}, shown in their Figure 6(c). Controls are applied until the final time $T=5$, and we also take parameters $s=0$ and $\gamma = 1$ for the cost functional. The case of zero forcing is used as the initial control guess, $\zeta^{(1)} = 0$, thus allowing a cost comparison between the uncontrolled and optimally controlled systems. The results of the forward--backward optimisation procedure (Algorithm \ref{alg:forback}) for $\kappa = -0.5$ are shown in panels (b--d) of Figure \ref{optimalExp2fig}. Panel (b) plots the difference $\eta^* -\overline{\eta}$ at the final time, which is visibly $O(1)$; the $L^{\infty}$-norm of this surface is $1.13$, with the corresponding value at $t=0$ being $3.73$, and taking a minimum value of approximately $1.10$ across the whole time interval. This large deviation is expected since the desired state is a non-solution to the 2D KSE \eqref{controlled2dks}, and the time interval over which we apply controls is relatively short. Panel (c) compares the cost functional integrands for the uncontrolled and optimally controlled cases. The uncontrolled solution, shown with the dotted line in panel (c), is seen to grow even over the course of $5$ time units, with cost $15.00$. The $L_0^2$-norms of $\zeta^*$ and $\eta^*-\overline{\eta}$ are decreasing for the majority of the time interval. The evolution of $\eta^*$, $\eta^* - \overline{\eta}$ and $\zeta^*$ for this case are presented in Supplementary Movie 1 available at \url{https://youtu.be/WczsshjrKW0}, and we note that the observed optimal control is not too far from a proportional control, i.e. $ - \zeta^* \sim \eta^* - \overline{\eta}$. The sequence of costs attained by the minimising sequence $(\eta^{(n)},\zeta^{(n)})$ is displayed in Figure \ref{optimalExp2fig}(d), including the breakdown into the three components $\mathcal{C}_{0,1}^{(j)}(n)$ for $j= 1,2,3$. The dotted lines with half integer values correspond to the cases where the updated control results in an increased value of the cost functional, and the step size is halved, as shown in the plot of $c^{(n)}$ (faded dashed line) taking values on the right axis. 
\begin{table}[tbhp]
  \caption{Comparison of results for a range of $\kappa$.}
  \label{tab:simpletable}
  \centering
  \begin{tabular}{|c|c|c|c|c|c|c|} \hline
   $\kappa$ & -0.5 & 0 & 0.5 & 1 & 1.5 \\ \hline
    $\mathcal{C}_{0,1}^{(1)}$ for $\zeta = 0$ &    12.42 & 12.27 & 12.23 & 12.22 & 12.21 \\ \hline
    $\mathcal{C}_{0,1}^{(2)}$ for $\zeta = 0$ &  2.58 & 2.46 & 2.44 & 2.44 & 2.44 \\ \hline
    $\mathcal{C}_{0,1}$ for $\zeta = 0$  &  15.00 & 14.73 & 14.68 & 14.66 & 14.65 \\ \hline
    $\mathcal{C}_{0,1}^{(1)*}$  & 2.01 & 3.91 & 6.07 & 7.78 &  8.97 \\ \hline
    $\mathcal{C}_{0,1}^{(2)*}$   & 0.18 & 0.56  & 0.95 & 1.24 & 1.46 \\ \hline
    $\mathcal{C}_{0,1}^{(3)*}$   & 2.40 & 3.19 & 2.95 & 2.40 & 1.89 \\ \hline
    $\mathcal{C}_{0,1}^*$  & 4.59 & 7.66 & 9.97 & 11.42 & 12.31\\ \hline
     $n(\tau=5\times10^{-6})$   & 225 & 68 & 45 & 46 & 48 \\ \hline
    $\min_{t\in[0,T]} \| \eta^* - \overline{\eta} \|_{L^{\infty}}$   & 1.09 & 1.90 & 2.32 & 2.58 & 2.74 \\ \hline
  \end{tabular}
\end{table}
The cost breakdowns for the uncontrolled and optimally controlled cases are given in Table \ref{tab:simpletable} for each $\kappa$. Rows $2$ to $4$ correspond to the uncontrolled case, with cost components $\mathcal{C}_{0,1}^{(1)}$, $\mathcal{C}_{0,1}^{(2)}$ (note that $\mathcal{C}_{0,1}^{(3)} = 0$ for $\zeta = 0$) and total cost $\mathcal{C}_{0,1}$ -- these are observed to be decreasing functions of $\kappa$. The next four rows give the optimal cost and its breakdown into its three components; these asymptote values were obtained by fitting the sequences of costs to a function of the form $a e^{bn} + c$. The increase in the optimal cost with $\kappa$ over these examples may be explained as follows: Since we start from small amplitude initial conditions and such a short time interval $[0,T]$ for optimisation, the linear instabilities work in favour of controlling the solution to the desired state; for $\kappa = 1.5$, the control must be strong since the zero solution is exponentially stable. However, we expect a turning point in this behaviour as $\kappa$ becomes very negative (dependent on $T$), when the cost of controlling the linear instabilities outweighs the control cost saved due to the linear instabilities aiding the solution growth towards $\overline{\eta}$. Row 9 in Table \ref{tab:simpletable} gives the number of iterations required to ensure that the change in the minimising sequence of costs is below the desired tolerance $\tau = 5 \times 10^{-6}$; this becomes very large as $\kappa$ decreases further beyond $-0.5$. Finally, the minimum of the $L^{\infty}$-norm of $\eta - \eta^*$ across the whole time interval is given in row 10 -- the optimal controls approach the desired state more closely in the $L^{\infty}$-sense as $\kappa$ decreases (for the range of numerical simulations we completed). We also performed numerical experiments over a range of $s$ and $\gamma$, although not shown here, and we report that the dependence of the optimal cost $\mathcal{C}_{s,\gamma}^*$ on these parameters was the same as in the linear case shown in Figure \ref{costfunscatter}, i.e. an decreasing function of $s$ and an increasing function of $\gamma$; the same trend was observed for the individual components, $\mathcal{C}_{s,\gamma}^{(j)*}$ for $j = 1,2,3$.

\begin{figure}[]
\centering
\begin{subfigure}{2.9in}
\caption{Minimising sequences of costs.} 
\includegraphics[width=2.9in]{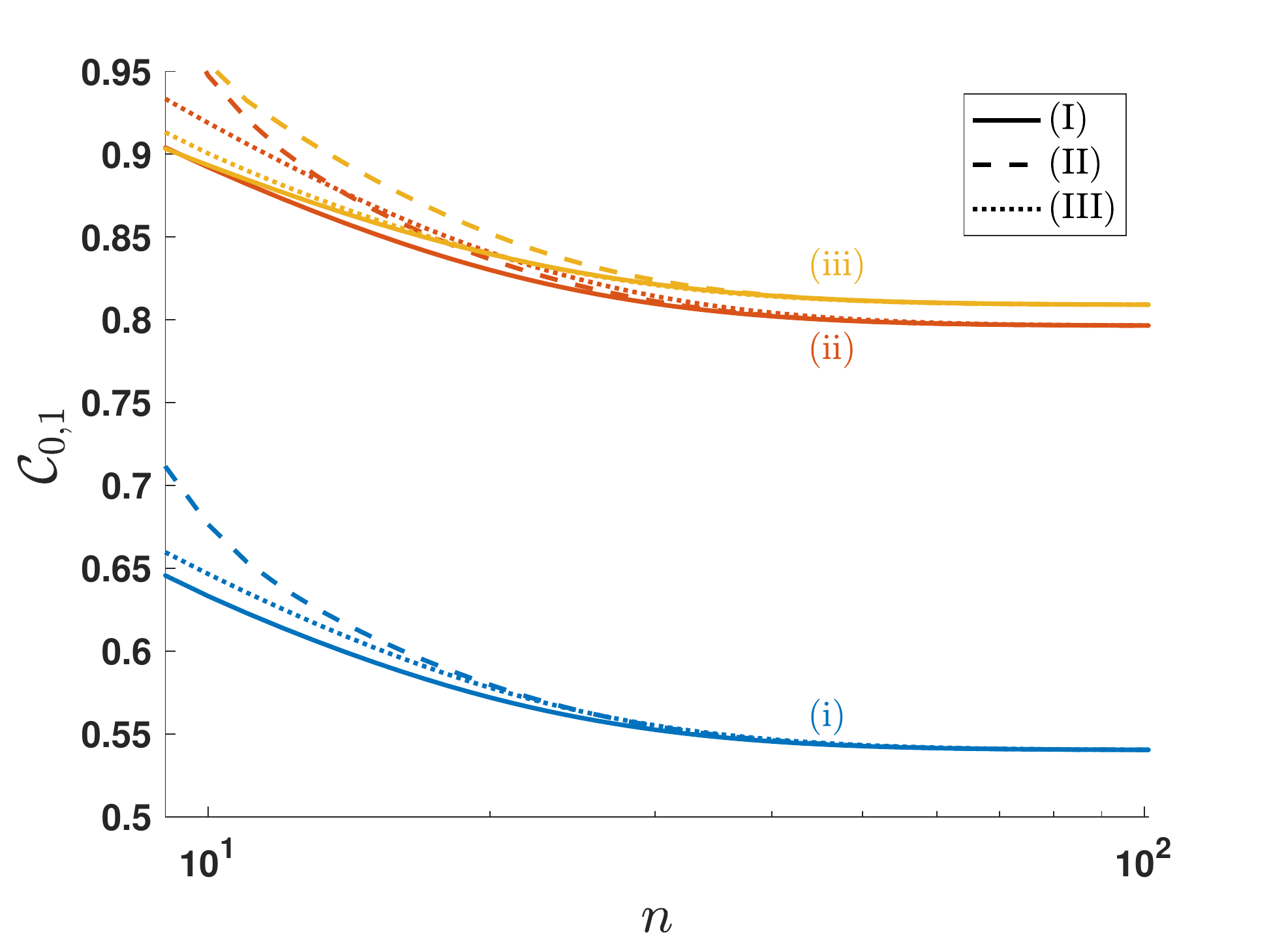}
\end{subfigure}
\begin{subfigure}{2.9in}
\caption{ Components of $\mathcal{C}_{0,1}^*$ for case (i).} 
\includegraphics[width=2.9in]{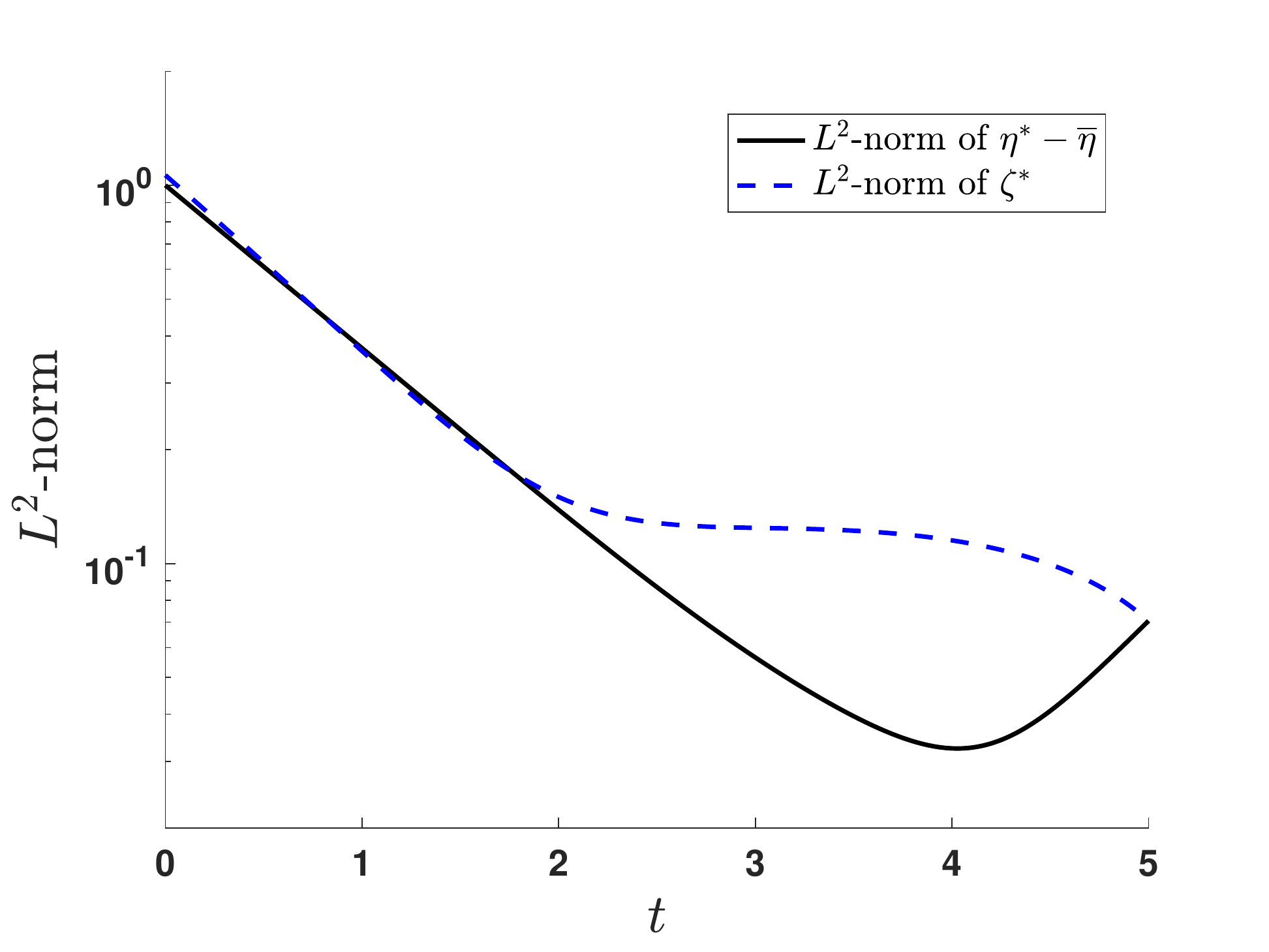}
\end{subfigure}
\begin{subfigure}{2.9in}
\caption{ Components of $\mathcal{C}_{0,1}^*$ for case (ii).} 
\includegraphics[width=2.9in]{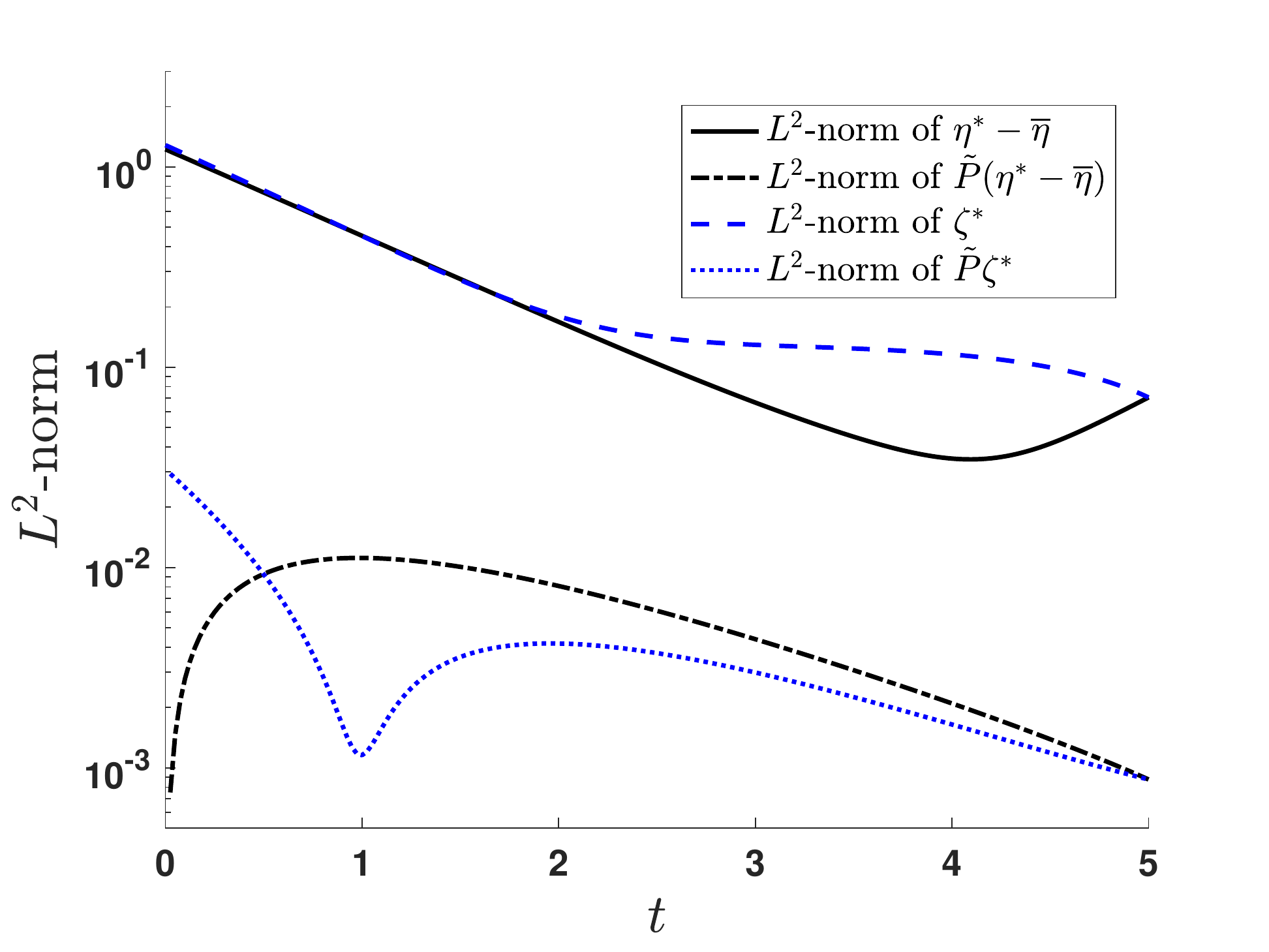}
\end{subfigure}
\begin{subfigure}{2.9in}
\caption{ Components of $\mathcal{C}_{0,1}^*$ for case (iii).} 
\includegraphics[width=2.9in]{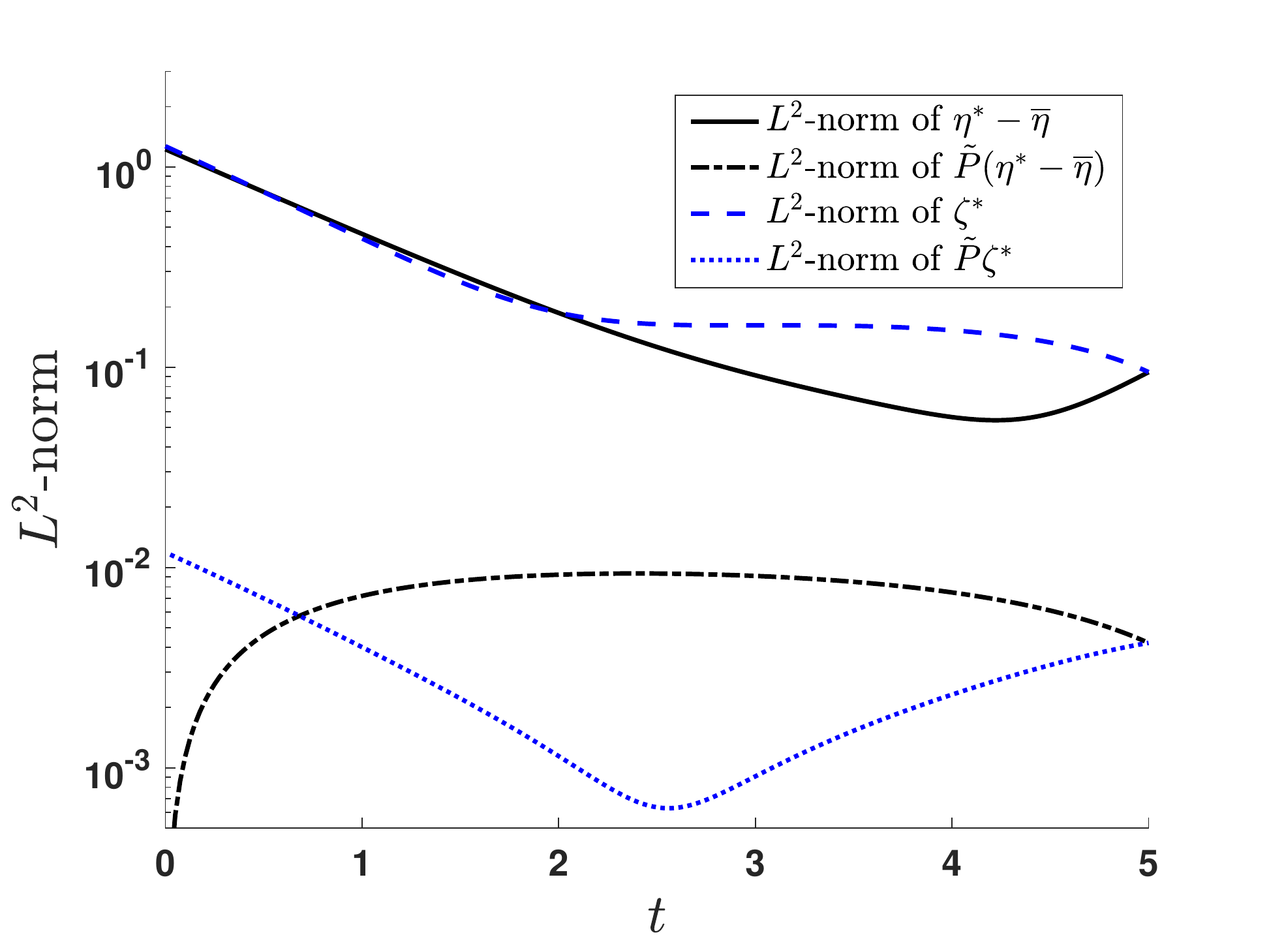}
\end{subfigure}
\caption{Full optimal control for a vertical film flow ($\kappa = 0$) and $L_1 = 32$, $L_2 = 21$ with $s=0$ and $\gamma = 1$ over the time interval $[0,5]$ for the range of initial conditions and desired states given in \eqref{init_desstate_FOC}.} \label{optimalExp3fig}
\end{figure}
For the second numerical experiment, we consider a vertical film ($\kappa = 0$) and lengths $L_1 = 32$ and $L_2 = 21$. We consider three cases, (i),(ii), and (iii), of initial conditions and time-dependent desired states,
\begin{equation}\label{init_desstate_FOC}v = \sin\left(\frac{4\pi x}{L_1} \right) + \delta_{\textrm{(ii)}}\cos\left(\frac{2\pi x}{L_1} + \frac{2\pi y}{L_2} \right), \quad \overline{\eta} =  \sin\left(\frac{2\pi (x - t)}{L_1} \right) + \delta_{\textrm{(iii)}} \sin\left(\frac{2\pi x}{L_1} + \frac{2\pi y}{L_2} \right),\end{equation}
where $\delta_{\textrm{(ii)}}$ and $\delta_{\textrm{(iii)}}$ are $1$ for their respective cases and zero otherwise.
In case (i), both initial condition and desired state are spatially 1D, whereas in cases (ii) and (iii), a mixed-mode term is added to either $v$ or $\overline{\eta}$, respectively. No transverse modes are included in $v$ or $\overline{\eta}$ in any of the cases, neither are these modes linearly unstable. We additionally consider a variety of initial control guesses,
\begin{equation}\label{threechoiceszeta1}\zeta^{(1)}_{\textrm{(I)}} =  0 , \quad  \zeta^{(1)}_{\textrm{(II)}} =  \cos\left(\frac{2\pi y}{L_2} \right), \quad \zeta^{(1)}_{\textrm{(III)}} =  \sin\left(\frac{2\pi x}{L_1} + \frac{2\pi y}{L_2} \right).
\end{equation}
We also fix the final time $T=5$, and set control parameters $s=0$ and $\gamma =1$ as in the previous numerical experiment. In Figure \ref{optimalExp3fig}(a), we show a portion of the minimising sequence of costs for the three cases of $(v,\overline{\eta})$, each with the three choices of $\zeta^{(1)}$ given in \eqref{threechoiceszeta1}. The limits of the sequences are independent of the initial control guess, and we confirm from inspection of the numerical solution that the approximations of $\zeta^*$ also agree (this is not necessarily implied by the previous statement). This lends credence to the possibility that the obtained optimisers are global minima of the cost functionals. Figure \ref{optimalExp3fig}(b,c,d) plots the $L_0^2$-norms of the optimal control and state over $[0,T]$ for the three cases. For case (i), both $\zeta^*$ (and $\eta^*$) remain 1D for the entire time interval; the optimal control obtained here is the same that would be obtained from the 1D simplification of the control problem. In cases (ii) and (iii), two-dimensionality enters into the problem via the initial condition and desired state, respectively, through the addition of a mixed mode term. We find that not only mixed modes appear in $\zeta^*$, but the projection onto the transverse modes is also non-trivial. For these two cases, the $L_0^2$-norms of the projections $\tilde{P}(\eta^* - \overline{\eta})$ and $\tilde{P}\zeta^*$ are included in Figure \ref{optimalExp3fig}(c,d), and visualisations of the time evolution are given in Supplementary Movies 2 (available at \url{https://youtu.be/V24mi-C8n1g}) and 3 (available at \url{https://youtu.be/hfcmplP83fQ}), respectively. This can be understood with the adjoint equation \eqref{Adjointeqn1}, where it can be seen that mixed mode and streamwise mode activity in the the solution excites the purely transverse modes in the adjoint variable. With these results, we make the following conjecture about the spatial dimension of the optimal state and control pair:
\begin{conj}\label{conj1}Assume that $F_{\textrm{ad}}$ cannot be spanned by functions of $x$ and $t$ alone ($F_{\textrm{ad}}$ contains functions dependent on $y$). The optimal control $\zeta^*$ (and $\eta^*$) is independent of $y$ if and only if $v$ and $\overline{\eta}$ are independent of $y$. Moreover, if $F_{\textrm{ad}}$ contains functions which are independent of $x$, i.e. transverse modes, and $\zeta^*$ is dependent on $y$, then $\tilde{P}\zeta^*$ is non-zero.
\end{conj}
The presence of transverse modes in the optimal control even if none are present in $v$ or $\overline{\eta}$ may at first appear unusual, the system of transverse modes decouples partially, and the transverse modes in the optimal control influence the dynamics of the streamwise and mixed modes through the nonlinearity. We investigate the effect of transverse mode forcing in more detail in the next section, analysing the response of the interface energy to spanwise blowing and suction patterns. We also note that, as before, the optimal states and controls appear to be close to proportional for the majority of $[0,T]$ (see Supplementary Movies 2 and 3).

In the situation when $F_{\textrm{ad}}$ is a strict subset of $L^2(0,T;L_0^2)$, the numerical procedure proceeds exactly as outlined above, but with projections of the updated controls onto $F_{\textrm{ad}}$ at each iteration; this is viable due to the linearity of the adjoint equation \eqref{Adjointeqn1} in $p$. We note that both Theorem \ref{existenceunThm1} and Theorem \ref{thm4.2} are valid without the zero-average restriction, i.e. taking $F_{\textrm{ad}} \subseteq L^2(0,T;L^2)$, where the latter space is given an appropriate norm. We repeated a number of the above numerical experiments in this setting, allowing controls in $L^2(0,T;L^2)$. In this case, the projection $P_{\bm{0}}$ is removed from the adjoint equation. We found that the optimal controls in this space caused large drifts in the spatial average of the solution. Physically, this would require a large reservoir of fluid, and drastic modification of the average film height could result in dewetting or leaving the thin film regime altogether.

\section{Transverse mode effects\label{TranwaveeffectsSec}}

In this section, we briefly examine the extent to which controlled transverse modes can affect the streamwise and mixed modes through the nonlinear coupling -- see the ODE system of Fourier coefficients \eqref{fouriercoeffsfull1}. In this way, purely transverse controls (such as those studied in section \ref{opttranssubsec1}) may be thought of as indirect controls on the full system. This may be useful in physical situations where it is easier to force (and/or observe) the transverse modes than the other components of the flow. Furthermore, at the weakly nonlinear level, the transverse component of the problem is linear and diagonal, thus can be broken down into a set of 1D ODEs, and its control is well-studied (see section \ref{opttranssubsec1}). 

We focus on the vertical film setup, $\kappa = 0$, which has been studied extensively in \cite{tkp2018}. In this case, the transverse modes of uncontrolled solutions are damped for any choice of $Q$ with chaotic dynamics emerging for sufficiently large $L_1$ and $L_2$. In \cite{tkp2018}, the authors found that the time-averaged energy density behaves as
\begin{equation}\label{KS2Dtkpbound} \llangle \eta \rrangle := \lim_{T \rightarrow \infty} T^{-1/2} \| \eta \|_{L^2(0,T;L_0^2)}  \approx 1 ,
\end{equation}
for sufficiently large length scales; the quantity $E_{L,\alpha}$ considered in \cite{tkp2018} is related to $\llangle \eta \rrangle^2$ by the factor of $L_1^{-1}L_2^{-1}$ due to our definition of $L_0^2$-norm (\ref{Hsnorm1aaa}a). The estimate \eqref{KS2Dtkpbound} is additionally found numerically for the 1D KSE \eqref{KSEintro}, although current analytical bounds are not sharp. The goal of this section is to determine if purely transverse controls may be used to decrease $\llangle \eta \rrangle$ below $1$, i.e. make the fluid interface less energetic on average. Another appealing property of a controlled system would be the regularisation of chaos.

The appropriate choice of desired state for this study is $\overline{\eta} = (I - \tilde{P}_{\Sigma}) \eta + \tilde{\psi}(y,t)$, where $\Sigma$ is the set of transverse Fourier modes which we can force ($\tilde{P}_{\Sigma}$ is a projection onto these modes) and $\tilde{\psi}$ is a real-valued function of the modes in $\Sigma$ alone, with the property that $\tilde{\psi}_{-k_2}$ is the complex conjugate of $\tilde{\psi}_{k_2}$. Transverse modes which are not included in $\Sigma$ decay exponentially for the vertical film case under consideration here, thus we do not take them in our initial conditions to prevent any transient effects. The derivation of optimal controls for this problem is provided in appendix \ref{appendixdericoptlinear}, generalising the controls used in section \ref{opttranssubsec1}. Rather than specifying any particular control, we assume (sub-optimal) controls are applied so that the desired state is achieved exactly for all time, i.e. $\tilde{\eta} = \tilde{\psi}$. This is reasonable since we have full reachability and controllability for the individual transverse modes; the ODEs \eqref{Fouriercoefftranssystem1earlier} allow explicit construction of a control $\tilde{\zeta}$ for a given state $\tilde{\eta}$ which varies continuously from $\tilde{P}v$ to $\tilde{\psi}$. With this, we may focus on the flow response to transverse modes with a fixed amplitude.

\begin{figure}
\centering
\begin{subfigure}{2.9in}
\caption{Time-averaged energy density of $\eta - \tilde{\eta}$.} 
\includegraphics[width=2.9in]{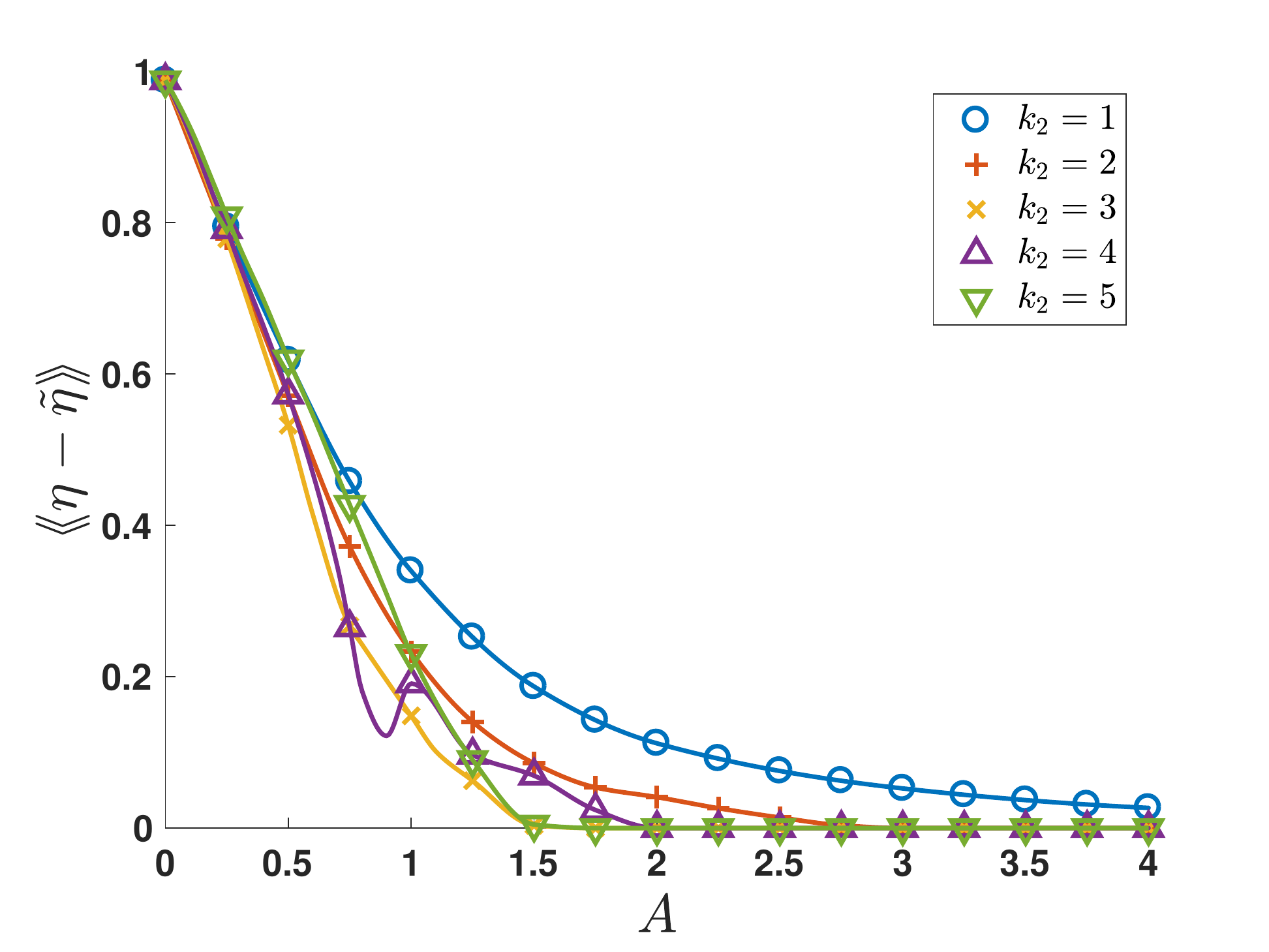}
\end{subfigure}
\begin{subfigure}{2.9in}
\caption{Time-averaged energy density of $\eta$.} 
\includegraphics[width=2.9in]{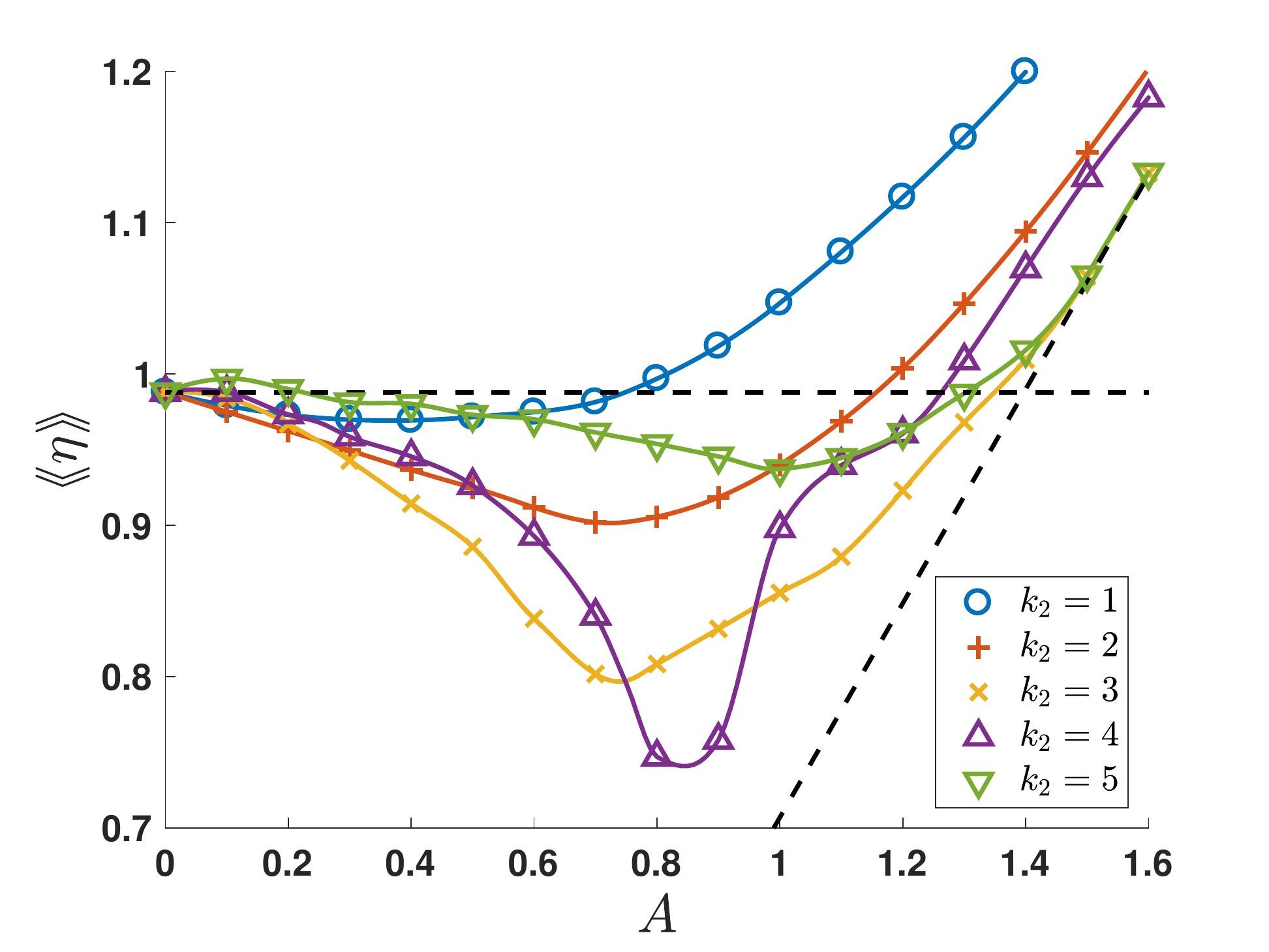}
\end{subfigure}
\caption{Effect of transverse waves on the time-averaged energy density of the full solution and its projection onto the streamwise and mixed modes. Panel (a) shows the time-averaged energy in the streamwise and mixed modes plotted against the amplitude $A$ for a range of $k_2$, and panel (b) considers the full solution. In the latter, the horizontal dashed line corresponds to the average energy in the case of $A=0$, and the diagonal dashed line corresponds to the energy in the transverse wave, $A/\sqrt{2}$. The difference between a line in panel (b) and the diagonal dashed line gives the corresponding line in panel (a). Note that the markers are used to differentiate between the different cases, and do not correspond directly to the datapoints from numerical simulations.}\label{nontrivialoptdesiredstate}
\end{figure}

We consider the simple choice of $\tilde{\psi} = A\sin( \tilde{k}_2 y )$ for an amplitude $A\geq 0$ and transverse wavenumber $\tilde{k}_2$, only containing contributions from the $(0,\pm k_2)$-modes. In our numerical simulations, we fix domain lengths $L_1 = 120$, $L_2 = 30$, in which case the uncontrolled system evolves chaotically. Additionally, we take random (real-valued) initial conditions with unstable streamwise and mixed modes so that the solution enters the global attractor rapidly. As noted above, we do not include transverse modes other than $\pm k_2$ in the initial condition, a choice which is justified by the numerical results themselves. The time-averaged energy density for both $\eta - \tilde{\eta}$ and $\eta$ are plotted in Figure \ref{nontrivialoptdesiredstate} for $k_2 = 1,\ldots,5$ and a range of $A$. These time-averages are approximated by averages over a large time interval $[T_1,T_2]$ after the solution has entered the global attractor (see \cite{tkp2018} for the details). Panel (a) shows that an increase in $A$ results in decay of the energy density of $\eta - \tilde{\eta}$ (this decrease is monotonic in most cases). Furthermore, we find that the chaotic dynamics are eventually regularised for a sufficiently large value of $A$, and the transverse wave eventually becomes nonlinearly stable where the line in Figure \ref{nontrivialoptdesiredstate}(a) touches down on the $A$-axis; these critical values and the behaviours of the energy densities are heavily dependent on $k_2$. For $k_2 = 1$, the chaotic dynamics are regularised for $A \approx 5.8$ and the transverse sine-wave becomes nonlinearly stable at $A=10.22$. For $k_2 = 2$, this happens for much lower amplitudes, with chaos being regularised for $A$ just beyond $2$, and the trivial solution in the streamwise and mixed modes, $\eta - \tilde{\eta} = 0$, becoming stable at $A =  2.86$. Panel (b) shows the more surprising result that a small amplitude transverse wave can lower the total average energy of the system (however, this does not account for any control costs). Due to the form of the nonlinear coupling, we postulate that very high frequency transverse waves will have little effect on the system energy and dynamics (as is confirmed by linear results to follow). We find that the choices of $k_2 = 3, 4$ are particularly successful in decreasing $\llangle \eta \rrangle$; these results suggest the existence of an optimal mode of control and amplitude (if $\tilde{\psi}$ is restricted to a steady modal wave) to minimise the system energy, dependent on $\kappa$, $L_1$ and $L_2$. From our time-dependent simulations, we also observed that larger values of $A$ also delayed the onset of chaos, below the threshold of regularised dynamics.

\begin{figure}
\centering
\includegraphics[width=2.9in]{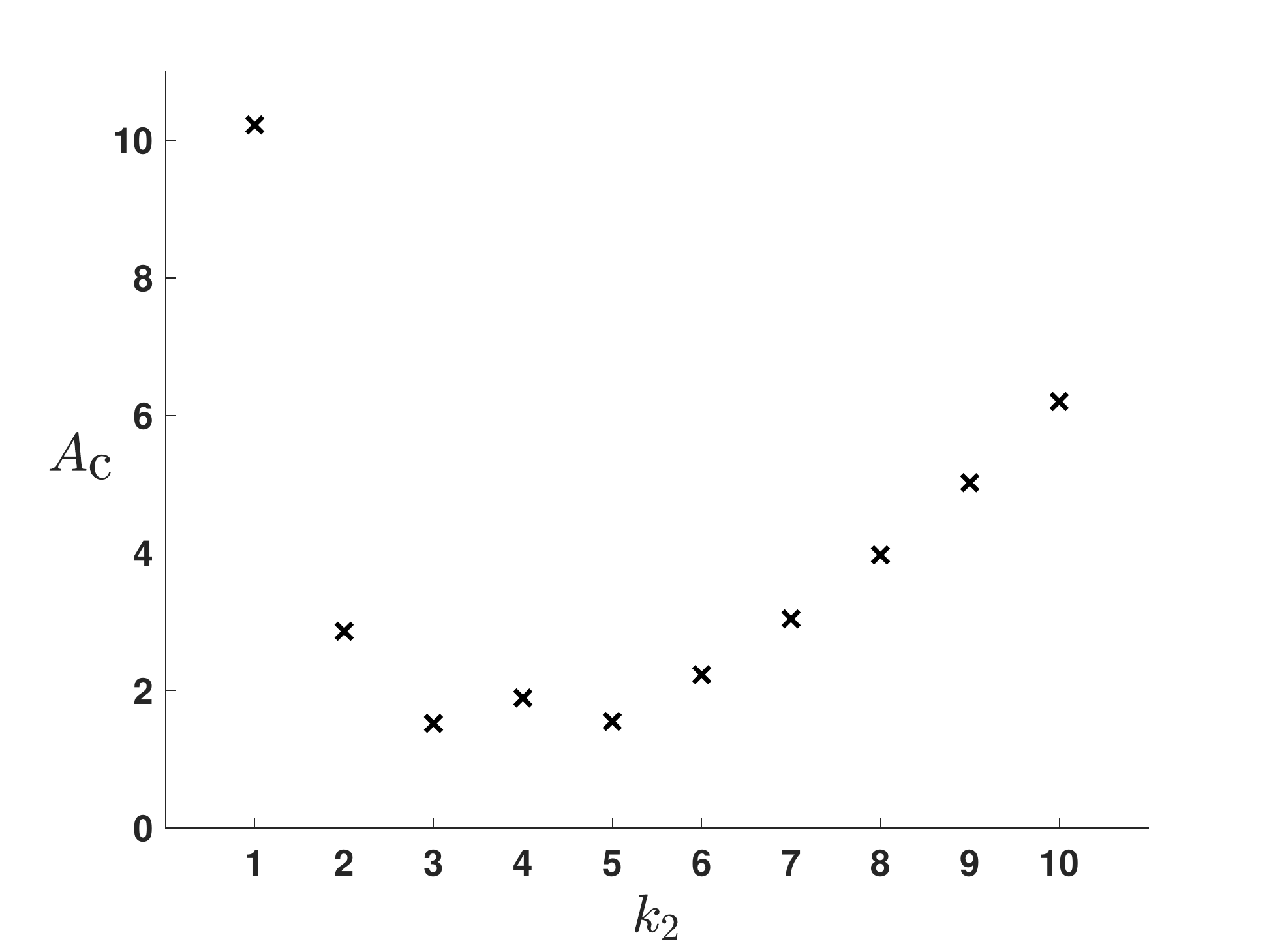}
\caption{Critical amplitude $A_{\textrm{c}}$ which separates the linearly stable and unstable regimes plotted against $k_2$.
}\label{A_crit_vs_k2}
\end{figure}
The numerical results in Figure \ref{nontrivialoptdesiredstate}(a) show that full nonlinear stability of the transverse wave occurs at the lowest values of $A$ for $k_2 = 3,5$. We now use linear theory to investigate the relationship between this critical value $A_{\textrm{c}}$ and $k_2$. In terms of the infinite-dimensional system of Fourier modes, the linearisation of \eqref{fouriercoeffsfull1} about $\eta = \tilde{\eta} = \tilde{\psi}$ for a general $\tilde{\psi}(y,t)$ is
\begin{equation}\label{fouriercoeffsfull23}\frac{\mathrm{d}}{\mathrm{d}t}\eta_{\bm{m}} = -  i\tilde{m}_1 \sum_{l_2\in\mathbb{Z}} \eta_{(m_1, m_2 - l_2)} \tilde{\psi}_{l_2} + \left[ (1-\kappa) \tilde{m}_1^2 - \kappa \tilde{m}_2^2 - |\bm{\tilde{m}}|^4  \right] \eta_{\bm{m}},\end{equation}
for $\bm{m} \in \mathbb{Z}^2$ with $m_1 \neq 0$, which, for the case of $\kappa = 0$ and $\tilde{\psi} = A\sin( \tilde{k}_2 y )$ becomes
\begin{equation}\label{fouriercoeffsfull24}\frac{\mathrm{d}}{\mathrm{d}t}\eta_{\bm{m}} = \frac{A \tilde{m}_1}{2}  [\eta_{(m_1, m_2 + k_2)} - \eta_{(m_1, m_2 - k_2)} ]+ \left[  \tilde{m}_1^2 - |\bm{\tilde{m}}|^4  \right] \eta_{\bm{m}}.\end{equation}
The stability of the above infinite-dimensional linear system can be determined by considering the stability of sufficiently large finite-dimensional truncation. With this, we may compute the critical value $A_{\textrm{c}}$, which separates the unstable and stable regimes, as a function of $k_2$. The results are plotted in Figure \ref{A_crit_vs_k2}, and the linear theory agrees with the critical values of the amplitude found in our nonlinear time-dependent simulations when the transverse wave becomes attractive to all initial conditions. Although not plotted, $A_{\textrm{c}}$ for $k_2 \geq 10$ is monotonically increasing, as suggested by the form of the nonlinearity. Although not as informative as Figure \ref{nontrivialoptdesiredstate}, the results shown in Figure \ref{A_crit_vs_k2} are a good predictor to whether the flow has a weak or strong response to a given frequency.

The consideration of more general $\tilde{\psi}$ and different parameter choices is beyond the scope of our study. Transverse controls $\tilde{\zeta}$ allow us to attain any such transverse wave state (which then may be thought of as a control in its own right), however, the reachability and controllability of the streamwise and mixed mode system for the control $\tilde{\psi}$ is unclear, in which the control acts multiplicatively through the nonlinear coupling, rather than additively. We believe that these results may have applications to other control problems for multidimensional flows with a dominant direction, for example in aerodynamics.

\section{Conclusions\label{CONC1}}

In this paper we considered the distributed control of a Kuramoto--Sivashinsky equation for gravity-driven thin film flow overlying or hanging from a 2D flat substrate. Blowing and suction controls applied at the substrate surface appear as a forcing term in the weakly-nonlinear evolution equation. 

For hanging films ($\kappa < 0$), optimal controls (which are constant in the streamwise direction) were constructed in section \ref{opttranssubsec1} to impede the exponential growth of linearly unstable transverse modes; these controls were applied successfully in numerical simulations. This spanwise instability is physical, predicting the formation of rivulets which may be a precursor to dripping for certain parameter regimes. In a non-idealised situation, it may be much more difficult to prevent such an instability from developing. It may also be the case that the path to dripping may take a different route upon the control of this initial instability, bypassing rivulet formation -- will controls merely delay the onset of dripping? It is clear that weakly nonlinear analysis alone will not give an acceptable answer as the processes of rivulet formation and dripping are inherently nonlinear. The use of stabilising electric fields to prevent the dripping instability (blow-up) of the rivulets is currently under investigation by the authors. Furthermore, it may be viable to use the constructed optimal controls to develop a feedback controls for the manifestation of the same spanwise instability in the models higher up the heirarchy, for example the Benney equation \eqref{benneyqequation1}. This instability appears nonlinearly in these models, and thus the explicit construction of optimal controls is not possible, and iteration method are computationally expensive.

In section \ref{FOCref}, we considered a more general class of controls which varied spatially in both the streamwise and spanwise directions, yet restricted to zero spatial average; without this assumption the spatial average of the solutions was seen to vary greatly. A detailed proof of existence of an optimal control was given, outlining the general strategy that may be used for similar problems. Using the adjoint formulation, we constructed a forward--backward sweeping algorithm for the problem, and successfully applied it for numerical experiments. The optimal controls for the problem under consideration are initial condition (and parameter) dependent, a large number of iterations are required for the numerical algorithm to converge to the optimiser, and a large amount of data must be stored (unless checkpointing is used, which slows the algorithm). Thus, it is difficult to construct a (near) optimal control in real time if the problem is nonlinear and multidimensional, as in our case, and so it is unfeasible to use such controls in applications. However, the computed optimal controls indicate the efficacy of proportional controls, where $\zeta = - \alpha (\eta - \overline{\eta})$; a study of point-actuated controls by the authors will be presented elsewhere.

In section \ref{TranwaveeffectsSec}, we returned to the study of purely spanwise controls in a non-optimal setting. We focused on the vertical film case, for which controls are not required to ensure bounded solutions. With extensive numerical simulations, we found that small amplitude sinusoidal transverse waves yielded a decrease in the average energy of the fluid interface, with large amplitude waves being nonlinearly stable, regularising the chaotic dynamics.

A large motivation for this paper and current work by the authors is in the construction of controls for other systems in the hierarchy of models, with the ultimate goal of using control methodologies derived for the KSE to control the Navier--Stokes equations; from inducing recirculation regions to improve heat transfer to stabilising the exact Nusselt solution and preventing dripping for hanging film arrangements. This will also serve as a test for the weakly nonlinear models in describing the dynamics of the full system.

%

%
%
%
%
%
%
%
%

\appendix

\section{Derivation of optimal transverse control\label{appendixdericoptlinear}}

We now give a brief derivation of the optimal control for an LQR tracker problem for the transverse system \eqref{transversesystem1}. We consider the general desired state
\begin{equation}\overline{\eta} = (I -\tilde{P}_{\Sigma}) \eta + \tilde{\psi}(y,t), \qquad \tilde{\psi}(y,t) = \sum_{k_2 \in \Sigma} \tilde{\psi}_{k_2}(t) e^{i\tilde{k}_2 y}.\end{equation}
Here, $\tilde{P}_{\Sigma}$ is a projection onto a subset of the transverse Fourier modes, and $\tilde{\psi}$ is a real-valued function of the modes in $\Sigma$ alone, with the property that $\tilde{\psi}_{-k_2}$ is the complex conjugate of $\tilde{\psi}_{k_2}$; importantly then, $\Sigma$ must satisfy the property that if $k_2\in\Sigma$, then $-k_2\in\Sigma$. The term ``tracker" is used to indicate that the desired state is possibly non-trivial and time-dependent. In section \ref{opttranssubsec1}, this theory was employed for $\Sigma = \Xi$ and $\tilde{\psi} = 0$, and the full generalisation gives a viable choice of control to be used in conjunction with the results in section \ref{TranwaveeffectsSec}.

The linearity of the problem allows us to build optimal controls for each mode. From the $k_2^{\textrm{th}}$ component of the cost functional, we have the the Hamiltonian
\begin{equation}H_{k_2} =  |\tilde{k}_2|^{2s} \left|\tilde{\eta}_{k_2}(t) - \tilde{\psi}_{k_2}(t) \right|^2 + \gamma |\tilde{\zeta}_{k_2}(t)|^2 + p_{k_2} \left[-(\kappa \tilde{k}_2^2 + \tilde{k}_2^4) \tilde{\eta}_{k_2} +  \tilde{\zeta}_{k_2} \right]\end{equation}
where $p_{k_2}$ is the adjoint variable for the $k_2^{\textrm{th}}$ mode (with complex conjugate $p_{-k_2}$). Then, from Hamilton's equations, we have a two point boundary value problem
\begin{align} \label{optimalcontrol2pointbp1}
 \frac{\mathrm{d}}{\mathrm{d}t} \tilde{\eta}_{k_2}  = \frac{\partial H_{k_2}}{\partial p_{k_2}} & = - (\kappa \tilde{k}_2^2 + \tilde{k}_2^4) \tilde{\eta}_{k_2} +  \tilde{\zeta}_{k_2},  \\ \label{optimalcontrol2pointbp2}
 - \frac{\mathrm{d}}{\mathrm{d}t} p_{k_2}  = \frac{\partial H_{k_2}}{\partial \tilde{\eta}_{k_2}} & =  |\tilde{k}_2|^{2s} \left(\tilde{\eta}_{-k_2}(t) - \tilde{\psi}_{-k_2}(t) \right)- p_{k_2} (\kappa \tilde{k}_2^2 + \tilde{k}_2^4) ,  \\ \label{optimalcontrol2pointbp3}
 0 = \frac{\partial H_{k_2}}{\partial \tilde{\zeta}_{k_2}} & = \gamma  \tilde{\zeta}_{-k_2} + p_{k_2} \quad \Rightarrow \quad \tilde{\zeta}_{-k_2} = - \frac{p_{k_2}}{\gamma},
\end{align}
where the boundary conditions are
\begin{equation} \label{appenbc1} \tilde{\eta}_{k_2}(0) = v_{(0,k_2)} \qquad p_{k_2}(T) = |\tilde{k}_2|^{2s} \left(\tilde{\eta}_{-k_2}(T) - \tilde{\psi}_{-k_2}(T) \right).
\end{equation}
The final time boundary condition for the adjoint is found by differentiating the cost functional with respect to $\tilde{\eta}_{k_2}(T)$. Taking the complex conjugate of \eqref{optimalcontrol2pointbp2} and \eqref{optimalcontrol2pointbp3} gives the most useful form of the equations. To solve this two point boundary value problem, we make the ansatz that
\begin{equation}\label{ansatzappen1}p_{-k_2}(t) = - \gamma r_{k_2}(t)\tilde{\eta}_{k_2}(t) +  q_{k_2}(t).\end{equation}
Taking the time derivative of this and equating with the complex conjugate of \eqref{optimalcontrol2pointbp2}, after manipulations we arrive at
\begin{align}& \gamma\left[  -   \frac{\mathrm{d}}{\mathrm{d}t}r_{k_2}  - r_{k_2}^2 + 2 r_{k_2} (\kappa \tilde{k}_2^2 + \tilde{k}_2^4)   + \frac{|\tilde{k}_2|^{2s}}{\gamma} \right] \tilde{\eta}_{k_2} \nonumber \\
& \qquad\qquad\qquad\qquad\qquad\qquad\qquad\qquad = \left[ - \frac{\mathrm{d}}{\mathrm{d}t}q_{k_2}  + (\kappa \tilde{k}_2^2 + \tilde{k}_2^4 - r_{k_2}) q_{k_2} + |\tilde{k}_2|^{2s} \tilde{\psi}_{k_2}   \right]. \end{align}
Then we choose $r_{k_2}$ to satisfy \eqref{ricattieqn1}, and $q_{k_2}$ to satisfy
\begin{equation}\frac{\mathrm{d}}{\mathrm{d}t}q_{k_2} =  (\kappa \tilde{k}_2^2 + \tilde{k}_2^4 -  r_{k_2}) q_{k_2}  + |\tilde{k}_2|^{2s} \tilde{\psi}_{k_2} ,\qquad q_{k_2}(T) =  - |\tilde{k}_2|^{2s} \tilde{\psi}_{k_2}(T)\end{equation}
where the final time boundary condition is deduced from \eqref{appenbc1} and \eqref{ansatzappen1}. With these, the optimal control is given by $\tilde{\zeta}_{k_2}^* = r_{k_2}\tilde{\eta}_{k_2}^* - q_{k_2}/\gamma.$

\section{Estimate for proof of existence of optimal control\label{Appendixboundderiv}} 

Here we give a derivation of inequality \eqref{importantbound} used in section \ref{FOCref}. Multiplying \eqref{controlled2dks} by $\eta$ and taking the spatial average gives the energy equation
\begin{equation}\label{appB1} \frac{1}{2} \frac{\mathrm{d}}{\mathrm{d}t} \| \eta \|_{L_0^2}^2 = (1-\kappa)  \| \eta_x \|_{L_0^2}^2 - \kappa \| \eta_y \|_{L_0^2}^2 - \| \eta \|_{H_0^2}^2 + \langle \zeta, \eta \rangle_{L_0^2}
\end{equation}
where we have used integration by parts. Furthermore, we may bound
\begin{equation} (1-\kappa)  \| \eta_x \|_{L_0^2}^2 - \kappa \| \eta_y \|_{L_0^2}^2 - \frac{1}{2} \| \eta \|_{H_0^2}^2 =  \sum_{ \bm{k} \in\mathbb{Z}^2\backslash\{\bm{0}\} } \left[ (1-\kappa) \tilde{k}_1^2 - \kappa \tilde{k}_2^2 - \frac{1}{2} | \bm{\tilde{k}} |^4 \right] | \eta_{\bm{k}} |^2 \\
 \leq  C_1 \| \eta \|_{L_0^2}^2, 
\end{equation}
where $C_1 = \mathbbm{1}_{ \kappa < 1 }(1-\kappa)^2/2$. With this estimate, and applications of H\"{o}lder's and Young's inequalities to the term involving the control, \eqref{appB1} yields
\begin{equation}\label{appB2}  \frac{\mathrm{d}}{\mathrm{d}t} \| \eta \|_{L_0^2}^2 \leq C_2 \| \eta \|_{L_0^2}^2 -  \| \eta \|_{H_0^2}^2 +  \| \zeta \|_{L_0^2}^2,
\end{equation}
where $C_2 = 2 C_1 + 1$. Then, using Gronwall's inequality, we may compute
\begin{align} & \frac{\mathrm{d}}{\mathrm{d}t} \left( \| \eta \|_{L_0^2}^2 e^{-C_2 t} \right) \leq  -  \| \eta \|_{H_0^2}^2 e^{-C_2 t} +  \| \zeta \|_{L_0^2}^2 e^{-C_2 t}, \nonumber \\
\Rightarrow\quad & \| \eta(t) \|_{L_0^2}^2 e^{-C_2 t}  \leq  \| v \|_{L_0^2}^2 + \int_0^t -  \| \eta(t') \|_{H_0^2}^2 e^{-C_2 t'} +  \| \zeta(t') \|_{L_0^2}^2 e^{-C_2 t'} \; \mathrm{d}t', \nonumber \\
\Rightarrow\quad & \| \eta(t) \|_{L_0^2}^2 + e^{-C_2 T}  \| \eta \|_{L^2(0,t;H_0^2)}^2   \leq  e^{C_2 T} ( \| v \|_{L_0^2}^2  +  \| \zeta \|_{L^2(0,T;L_0^2)}^2).
\end{align}
Note the dependence on $t$ in both terms on the left hand side of the final line, and the uniform bound on the right hand side. From this, we may extract the desired bound
\begin{equation} \| \eta \|_{C^0([0,T]; L_0^2)} + \| \eta \|_{L^2(0,T;H_0^2)}  \leq C (  \| v \|_{L_0^2}  +  \|  \zeta \|_{L^2(0,T; L_0^2)}),\end{equation}
where $C = e^{C_2 T/2} + e^{C_2 T}$.

%
%

\section*{Acknowledgments} R.J.T. acknowledges the support of a PhD scholarship from EPSRC. The work of S.N.G was supported by EPSRC grants EP/K034154/1, EP/L020564/1, the work of D.T.P. was supported by EPSRC grants EP/K041134 and EP/L020564, and the work of G.A.P. was supported by EPSRC grants EP/L020564, EP/L025159 and EP/L024926.

\bibliographystyle{unsrt}
\bibliography{Controlbib}
\end{document}